\documentclass[num-refs]{wiley-networks}

\usepackage{graphicx,subfig}
\usepackage{epsfig}

\usepackage{ulem}
\usepackage{natbib}
\usepackage{amsmath, bm}
\usepackage{enumerate}
\usepackage{url}
\usepackage{float}
\usepackage{longtable}
\usepackage{booktabs}
\usepackage{lscape}
\usepackage{multirow}
\usepackage[ruled,vlined]{algorithm2e}
\usepackage[dvipsnames]{xcolor}
\usepackage{threeparttable}
\usepackage[breaklinks=true]{hyperref}
\usepackage{breakcites}

\usepackage{amsmath}
\usepackage{amssymb}
\usepackage{newtxmath} 
\usepackage{soul}

\usepackage{subeqnarray}
\usepackage{cases}
\usepackage{diagbox}
\usepackage{tikz}
\usetikzlibrary{arrows.meta,positioning}
\usetikzlibrary{calc}
\usetikzlibrary{decorations.pathreplacing}
\usepackage{subcaption}

\usepackage{makecell}

\usepackage{todonotes}

\newcommand{\REFONE}[1]{{\color{black}#1}}
\newcommand{\REFTWO}[1]{{\color{black}#1}}
\newcommand{\REFTHR}[1]{{\color{black}#1}}
\newcommand{\rev}[1]{{\color{black}#1}}
\newcommand{\REVII}[1]{{\color{black}#1}}  
\usepackage{bm}
\usepackage{booktabs}
\usepackage[T1]{fontenc}
\usepackage{titlesec}

\usepackage{xspace}

\newcommand{\prob}{\textsc{SCTP}\xspace}
\newcommand{\DBCG}{\textsc{DisBen-CG}\xspace}
\newcommand{\ABCG}{\textsc{AggBen-CG}\xspace}
\newcommand{\DBTP}{\textsc{DisBen-TP}\xspace}
\newcommand{\DET}{\textsc{Deter}\xspace}

\papertype{Original Article}
\paperfield{Networks} 

\title{Stochastic Corridor–Time Network Capacity Planning for Low-Altitude Airspace Systems}

\author[1]{Yipu Yao}
\author[1]{Li Ding}
\author[1]{Yanlu Zhao}

\affil[1]{Department of Management and Marketing, Durham University Business School, Durham University, Durham, DH1, 1SL, UK}

\corraddress{Yanlu Zhao}
\corremail{zhaoyanlu@durham.ac.uk}

\runningauthor{Yao et al.}

\begin{document}

\maketitle

\begin{abstract}
\REVII{Regulators in China, the United States, and the European Union now provide low-altitude airspace access as priced, time-windowed corridor authorizations, booked in advance and forfeited if unused. We ask how much capacity a UAV logistics planner should reserve on each corridor--time unit before demand is realized, to maximize expected profit net of reservation cost. Reserved capacity cannot be transferred across corridors or time windows and is consumed jointly along time-respecting paths, so reservations are coupled through the network in ways that models with exogenous airspace capacity cannot capture. We formulate a two-stage stochastic program whose recourse selects and routes accepted requests on a time-expanded network, prove its arc-based and path-packing forms equivalent, and solve it by Benders decomposition with column-generated subproblems. The decomposition operates on the LP relaxation, and all reported reservation and routing
decisions are recovered as integer plans. Computational experiments achieve single-digit LP-Benders gaps on moderate-sized networks and extend to much larger instances through a truncated-path approach. A Shenzhen case study shows reservations concentrating on structurally central corridors, with demand level and reservation price having more influence on the quantity of capacity reserved  than the selection of corridors.}

\keywords{Low-altitude airspace, Airspace capacity planning, Two-stage stochastic programming, Benders decomposition, Time-expanded network}
\end{abstract}


\section{Introduction}\label{sec:intro}

The low-altitude economy has emerged as a strategic component of next-generation transportation and digital infrastructure systems~\citep{guan2024exploration, huang2024low}. Advances in unmanned aerial vehicles (UAVs), autonomous navigation, and digital traffic management have enabled large-scale deployment of aerial services in urban environments, including logistics distribution, infrastructure inspection, emergency response, and public operations~\citep{mohsan2023unmanned, mohsan2022towards}. 
Traditionally treated as an open-access routing medium, low-altitude airspace is increasingly recognized as a scarce and structured resource~\citep{feng2025digital}.
As flight intensity increases, growing traffic density generates
measurable congestion risks and operational conflicts, making unmanaged access incompatible with safety and efficiency objectives~\citep{she2021efficiency,palmerius2024end}. Regulators are therefore redefining low-altitude airspace as managed infrastructure that must be systematically organized, allocated, and planned.

\REFTHR{This shift is already visible in industry practice. In Shenzhen, China, the Civil Aviation Administration of China (CAAC) Central-South Regional Bureau and the municipal Bureau of Transport have jointly opened more than $1{,}000$ approved low-altitude logistics routes, and operators such as Meituan UAV (urban food delivery), SF~Express and JD~Logistics (intracity parcel delivery), and EHang (passenger and cargo eVTOL trials) must apply for time-windowed corridor authorizations through the regional UAV traffic management platform before each flight~\citep{yang2026comprehensive}. Comparable corridor-and-window regimes are emerging internationally, including the FAA's Unmanned Aircraft System Traffic Management (UTM) framework in the United States, EASA's U-space service blocks in the European Union, and the JARUS SORA-based corridor pre-approvals adopted by Singapore and the United Arab Emirates ~\citep{henderson2025automation, nawaz2026towards, grote2022sharing}. 
Each grants access to a specific corridor during a specific time window, which makes the corridor--window pair the operationalized unit of accountability between operator and authority.
Three features of this regime drive everything that follows. Corridor authorizations are \textit{paid}, they are \textit{committed in advance of demand}, and unused reservations are typically \textit{non-refundable}. The operator therefore faces an irreversible investment in a resource whose value
is revealed only after demand materializes.}

The problem of allocating location--time resources in advance is not new to the operations literature. Airport slot allocation and air traffic flow management assign airport--time or sector--time capacity to flights under stochastic capacity realizations
\citep{bertsimas2016fairness,zografos2017increasing},
and stochastic service network design reserves capacity on a network in a first stage before demand-dependent flows are routed in a second
\citep{lium2009study}. Our setting differs from both in two respects. First, capacity is \textit{purchased endogenously} by the planner at a price, rather than rationed by a central authority across competing carriers, so the first-stage decision is an investment problem rather than an allocation problem. Second, the recourse combines \textit{admission control with path routing}: the planner selects which requests to serve as well as how to route them, whereas flow-management models retime a fixed and exogenous set of flights. The UAV planning literature, by contrast, optimizes routing, scheduling, or fleet deployment over airspace capacity that is specified exogenously \citep{chen2025hybrid,he2024distributed,li2022traffic,zhou2020resilient}. To our knowledge, the question of how much corridor--time capacity should be reserved in advance under stochastic demand, when that capacity is priced, non-transferable, and consumed jointly along routes, has received little attention.

\REVII{
Motivated by this gap, we study the Stochastic Corridor--Time Capacity Planning (\prob)  problem, in which a UAV logistics planner reserves corridor--time capacity in advance and, once uncertain delivery demand is realized, selects and routes the requests it will serve, with the objective of maximizing expected revenue from served requests net of reservation cost. The advance reservation decision involves a fundamental trade-off. Reserving more capacity mitigates congestion and service rejection during peak realizations but increases upfront cost, whereas reserving less reduces investment but forfeits profitable demand. What distinguishes {\prob} from a collection of independent newsvendor problems is that reserved capacity is neither freely reallocable nor independently valued: a corridor--time unit cannot be shifted to another corridor or another window once demand resolves, and its value depends on the availability of the other units that lie on the same time-respecting paths~\citep{van2025stochastic}.}

\REFTWO{This coupling manifests in three forms of congestion, and each shapes the structure of the recourse problem we must solve. \textit{Spatial} congestion arises because many requests compete for the same strategically located corridors connecting major demand centers, so a small number of corridor--time units act as persistent bottlenecks across scenarios. \textit{Temporal} congestion arises because demand is unevenly distributed over time and reserved capacity cannot be shifted across periods, so shortages during peak windows cannot be offset by slack elsewhere in the day. \textit{Propagation} congestion arises because each accepted request occupies capacity on several consecutive corridor--time units along its
route, so congestion at one location or time can propagate through downstream portions of the network and reduce the feasibility of other requests. 
These congestion mechanisms interact across space and time, creating strong dependencies
among requests and making the value of reserved capacity highly context dependent. As a result, effective planning requires jointly determining where and when capacity should be reserved while anticipating how future demand will compete for shared network resources.}





\begin{figure}[!h]
    \centering
    \includegraphics[width=\linewidth]{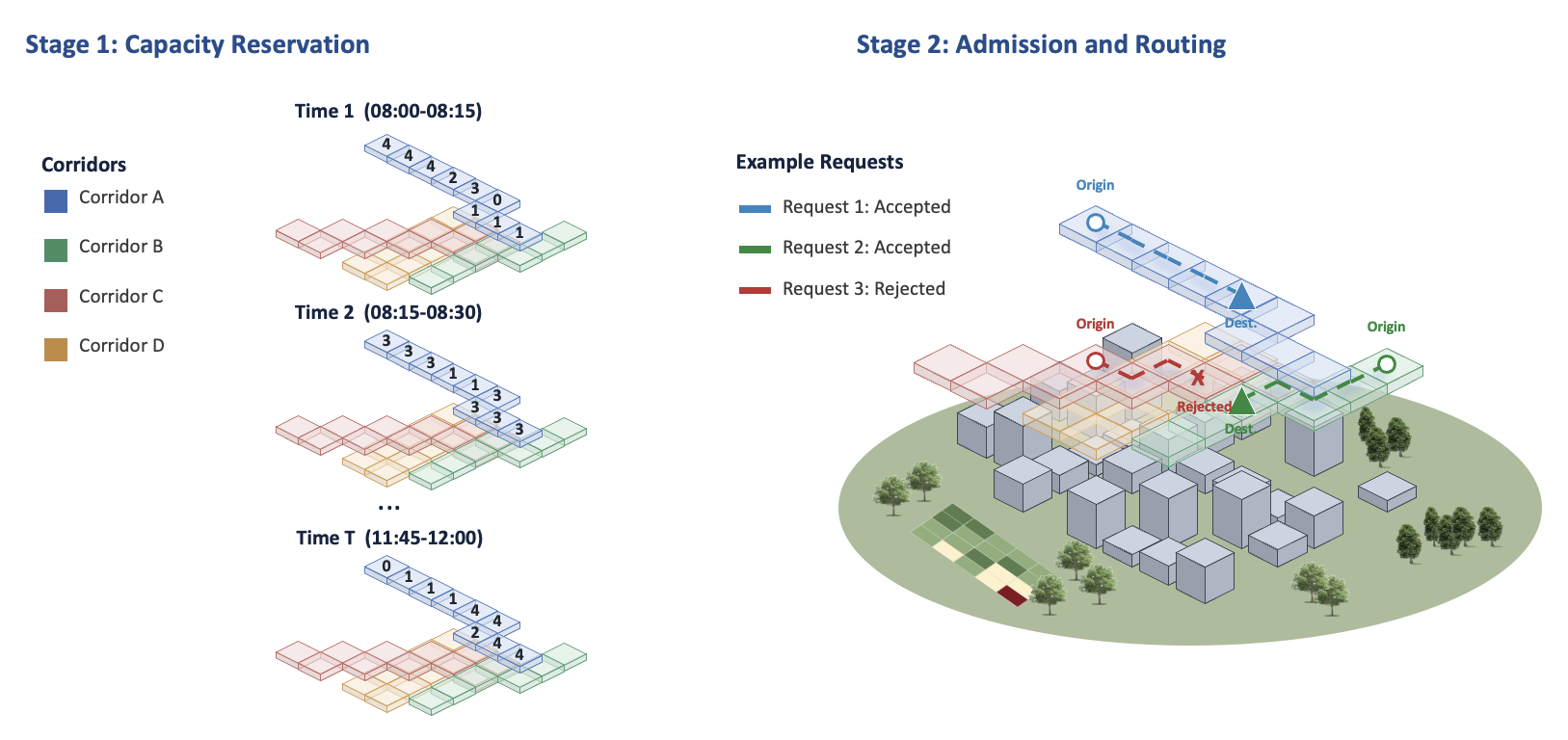}
    \caption{Two-stage corridor--time capacity planning under uncertain demand.}
    \label{fig:timeline}
\end{figure}

\REVII{These planning challenges naturally motivate a two-stage stochastic programming formulation~\citep{starita2020air}. 
Figure~\ref{fig:timeline} illustrates the proposed two-stage planning framework on a corridor network. The network consists of predefined flight corridors, each discretized into corridor--time capacity units over the planning horizon. In Stage 1, the planner reserves capacity for each corridor at each time interval before demand is realized, with the numbers indicating the reserved capacities. In Stage 2, realized requests are evaluated and routed through the network. The dashed lines show the trajectories of drone requests, which consume the corridor--time capacities reserved in Stage 1. Requests 1 (blue) and 2 (green) are accepted, whereas Request 3 (red) is rejected. 
We work on a time-expanded network rather than a static one because non-transferability across windows is a property of the resource, not of the schedule, and it cannot be represented by time-indexed capacities on a static graph without losing the time-respecting structure of feasible routes. We use a path-based rather than an arc-based recourse because the two are equivalent for this problem (Section~\ref{sec:model}) and the path-based form admits column generation, which avoids enumerating a path set that grows rapidly with the length of the planning horizon.
The model is deliberately positioned at the infrastructure-reservation layer, where the decision concerns airspace capacity rather than vehicle dispatch or fleet operations. Vehicle-level considerations therefore enter only through the construction of feasible candidate paths, while routing and demand acceptance constitute the operational recourse. The resulting formulation exhibits scenario-wise separability in the recourse, which we exploit through a disaggregated Benders decomposition
with column-generation subproblems.} This paper makes three contributions: 
\begin{itemize}
    \item \textbf{Relevance}. We introduce a stochastic corridor--time capacity planning framework in which low-altitude airspace capacity is modeled as a strategic, non-transferable spatio-temporal resource that the planner purchases rather than receives. The model integrates endogenous demand selection with path-based routing on a time-expanded network and shifts the focus from routing optimization under fixed infrastructure to infrastructure design under uncertainty.

    \item \textbf{Solution}. \REVII{We develop a stabilized Benders decomposition for stochastic corridor--time capacity planning in which each scenario recourse is solved exactly by column generation. The method attains single-digit LP-Benders gaps on instances with up to 25 nodes. For larger networks we report a truncated-path variant that extends the computational reach of the framework to 70 nodes together with valid multi-fidelity upper bounds. We quantify the price of this approximation explicitly rather than presenting it as exact: the truncated recourse systematically over-reserves capacity, and the large-instance results are therefore reported as a demonstration of computational reach rather than as recommended capacity plans.}
    \item \textbf{Insights}. \REVII{Extensive experiments provide several managerial insights into stochastic airspace capacity planning. First, ignoring demand uncertainty can forfeit $16.5$--$21.5\%$ of attainable profit, highlighting the value of incorporating stochastic demand information. Second, restricted routing flexibility is associated with systematic capacity over-reservation: planning and operating on a restricted path pool raises reservation expenditure by 26 to 67\% relative to the complete-flexibility regime, indicating that network flexibility substitutes for reserved capacity. Third, route overlap creates a fundamental trade-off: highly shared corridors increase the potential value of coordinated capacity planning because scarce resources can serve multiple requests, but they also intensify competition for limited corridor--time units. Finally, a Shenzhen case study shows that optimal reservations concentrate on structurally central corridors, while demand level and reservation price primarily determine the quantity of capacity reserved rather than the selection of corridors, yielding practical guidelines for congestion-sensitive airspace investment.}
\end{itemize}
The remainder of the paper is organized as follows. Section~\ref{sec:literature} reviews related work on airspace capacity planning and stochastic decomposition methods. Section~\ref{sec:model} formulates the stochastic corridor--time capacity planning problem and establishes the equivalence of its arc-based and path-packing formulations. Section~\ref{sec:analysis_algorithm} develops the proposed Benders decomposition and its column-generation subproblem solver. Section~\ref{sec:acceleration} introduces the acceleration strategies that make the decomposition converge at scale. Section~\ref{sec:numerical study} reports the computational study. Section~\ref{sec:case study} illustrates managerial implications in a geographically grounded setting. Section~\ref{sec:conclusion} concludes.

\medskip
\section{Literature Review}\label{sec:literature}


This work connects three streams of research. The first concerns the modelling of airspace as a structured and capacity-limited
infrastructure system, including strategic planning, routing, and
stochastic network design. The second concerns stochastic service network design and advance allocation of spatio-temporal resources. The third concerns decomposition methods for large-scale stochastic optimization, particularly Benders-type algorithms for network planning problems.

\paragraph{Structured Airspace Modelling and Capacity Planning.}
Recent research increasingly treats low-altitude airspace as a structured infrastructure rather than an open-access flight environment. 
Digital grid-based representations quantify operational risks and evaluate airspace capacity under urban constraints \citep{feng2025digital}. 
Three-dimensional traffic equilibrium models analyze congestion formation, system efficiency, and energy consumption in dense UAV operations \citep{she2021efficiency}. 
Analytical flow–density relationships for grid-like urban airspace networks demonstrate how local bottlenecks can destabilize overall system performance \citep{aarts2023capacity}. 
These studies establish that low-altitude airspace exhibits measurable and structural capacity limits. 
However, they focus primarily on operational evaluation and congestion analysis, rather than advance capacity reservation under uncertainty.

A related line of work studies the strategic airspace network design problem. 
Urban airspace networks can be constructed by selecting subsets of road corridors and projecting them into three-dimensional aerial structures subject to safety and technological constraints \citep{stuive2024airspace}. 
Strategic corridor layout problems have been formulated to balance connectivity and risk considerations in urban drone systems \citep{he2025air}. 
More broadly, multi-period stochastic network design models address evolving infrastructure in contexts such as humanitarian logistics and relief delivery \citep{liu2023multi}. 
In these models, network topology or facility configuration is optimized over time. 
By contrast, our setting assumes that the corridor network is given and focuses on determining how much capacity should be reserved on each corridor–time unit before demand is realized.

Another closely related stream studies routing, scheduling, and traffic control under fixed capacity constraints. 
Spatio-temporal UAV traffic scheduling models allocate routes and departure times to maintain separation and safety \citep{li2022traffic}. 
Hybrid centralized–decentralized frameworks coordinate airway assignment with three-dimensional trajectory design \citep{chen2025hybrid}. 
Distributed route planning introduces congestion pricing to coordinate competing origin–destination flows in dense airspace \citep{he2024distributed}. 
Queueing-based models analyze network stability and congestion propagation \citep{zhou2020resilient}, and fairness mechanisms for shared airspace have also been investigated \citep{merkert2021will,carraminana2025fair}. 
In these studies, airspace capacity is treated as exogenously given, and the analytical focus lies on operational allocation or control decisions.

This work also relates to stochastic routing and UAV logistics models that incorporate demand or operational uncertainty while keeping infrastructure fixed. In this regard,  two-stage stochastic routing with uncertain fuel consumption is studied in \citep{venkatachalam2018two}, and robust truck–UAV routing under uncertain demand is developed in \citep{faiz2024robust}. 
Demand-based UAV scheduling under airspace restrictions \citep{zhang2023research}, coordinated truck–drone routing \citep{di2021trucks}, dynamic vehicle routing with stochastic customers \citep{bent2004scenario}, and data-driven reserve allocation for drone delivery \citep{paul2025data} all optimize operational decisions within predetermined capacity limits. 
These models do not optimize advance corridor–time capacity reservation as a strategic first-stage decision.

Strategic capacity planning under uncertainty has a long tradition in air traffic flow management and infrastructure design. 
Deterministic and stochastic formulations allocate en-route capacity and manage sector delays under fixed sector constraints \citep{bertsimas1998air,lulli2007european}. 
Two-stage stochastic sector capacity planning integrates advance capacity ordering with downstream rerouting \citep{starita2020air}, and cross-border airspace capacity sharing under uncertainty has been examined \citep{kunnen2023cross}. 
Slot allocation models incorporate stochastic flying times via chance-constrained programming \citep{wang2023slot}. 
Outside aviation, stochastic facility location and distribution network design separate strategic infrastructure decisions from downstream allocation \citep{jayaraman2003simulated}, while multi-period capacity expansion models study long-term investment under demand uncertainty \citep{taghavi2016multi}. 
Vertiport location problems under behavioral uncertainty \citep{guo2025planning} and scenario blending techniques for infrastructure robustness \citep{kishore2026scenario} further extend this tradition.

\REFTHR{Across these studies, capacity is typically modeled as an aggregate arc- or facility-based resource. Conventional air traffic flow management (ATFM) provides a representative example, where sector capacity is defined as an aggregate hourly limit determined primarily by controller workload, and every flight traversing that sector during the corresponding period consumes one interchangeable unit of capacity~\citep{bertsimas1998air}. Our corridor-based planning problem differs from this aggregate capacity paradigm in two fundamental respects. First, capacity is indexed jointly by corridor and time period, so each reserved unit is tied to a specific corridor and time window and cannot be transferred across either dimension. This creates a much finer-grained representation of airspace resources than the aggregate sector-capacity models commonly adopted in conventional ATFM. Second, demand acceptance is endogenous because service requests may be rejected when reserved capacity is insufficient, and admission decisions are optimized jointly with route selection after demand realization. The operational problem couples capacity allocation and routing decisions much more tightly than conventional ATFM models, which primarily focus on balancing aggregate traffic flows~\citep{starita2020air}. Consequently, existing ATFM formulations cannot directly represent irreversible corridor–time reservation decisions.}

\REFONE{\paragraph{Positioning within Stochastic Service Network Design.}
The proposed framework is methodologically related to the literature on \textit{Stochastic Service Network Design} (SSND), which studies strategic transportation planning under demand uncertainty \citep{bai2014stochastic,moradi2025systematic}. In classical SSND, first-stage decisions determine a transportation service network before demand is observed, while second-stage decisions assign realized demand to the available services subject to capacity constraints \citep{crainic2014progressive,satici2026branch}. Similar to this setting, our framework separates strategic capacity planning from operational routing and evaluates first-stage decisions through stochastic recourse.

Despite these similarities, the proposed model differs from conventional SSND in several important respects. First, classical SSND typically determines which transportation services to operate or their operating frequencies \citep{crainic2000service}, whereas our framework reserves the amount of available capacity on each corridor over time. The resulting planning problem therefore focuses on allocating spatio-temporal airspace resources rather than selecting transportation services.
Second, demand fulfillment is endogenous. Rather than assuming that all realized demand must be transported, the planner jointly decides which requests to accept and how to route them after uncertainty is resolved. This introduces an explicit trade-off between reservation cost and operational revenue that is absent from most SSND formulations \citep{hewitt2022scheduled,wang2019stochastic}.
Third, airspace capacity is inherently indexed by both location and time. Capacity reserved for one corridor and time period cannot be transferred to another, creating a finer-grained resource allocation problem than the aggregate vehicle or service capacities commonly considered in freight transportation networks. Consequently, uncertainty primarily affects how limited corridor--time resources are distributed across competing requests.
Overall, the proposed model adapts ideas from the SSND paradigm by extending it from transportation service design to strategic airspace capacity reservation under demand uncertainty, while jointly optimizing demand acceptance and time-dependent routing.}

\paragraph{Decomposition Methods for Stochastic Programs.}

The second stream of literature concerns solution methodologies for large-scale stochastic programs. 
Benders decomposition is a classical approach for solving two-stage stochastic optimization problems by separating first-stage design decisions from scenario-specific operational subproblems \citep{rahmaniani2017benders}. 
Generalized Benders methods and branch-and-cut integrations enhance performance in mixed-integer settings \citep{lin2021branch}, and comprehensive surveys document classical, combinatorial, and generalized variants across network applications \citep{ibrahim2019bender}.

Multi-cut variants generate scenario-specific cuts instead of aggregating them into a single cut, often improving convergence in high-dimensional recourse problems \citep{you2013multicut,crainic2021partial}. 
Adaptive cut selection and stabilization techniques further enhance computational performance \citep{ramirez2023benders}. 
Extensions such as nested Benders and stochastic dual dynamic programming address multistage problems \citep{rebennack2016combining}, and time-consistent decomposition schemes have been proposed for distributionally robust stochastic optimization \citep{yu2021time}. 
In transportation and service network design contexts, branch-and-Benders-cut algorithms handle discrete first-stage design variables \citep{satici2026branch}, while progressive hedging provides an alternative scenario decomposition framework \citep{rajan2022routing}. 
These developments demonstrate the flexibility of decomposition-based methods in stochastic network applications.

While decomposition algorithms are well developed, most studies emphasize acceleration strategies, stabilization, or integration with branching schemes.
To the best of our knowledge, relatively limited attention has been given to how structural congestion geometry affects the convergence behavior of Benders decomposition in path-based spatio-temporal networks.
In our setting, route overlap determines how strongly different demand scenarios compete for the same corridor--time resources.
We empirically examine how this structural property affects the convergence performance of the proposed Benders decomposition in stochastic corridor–time capacity planning.

\medskip
\section{Model Development}\label{sec:model}
\REVII{This section formalizes the stochastic corridor--time capacity planning problem. Section~\ref{subsec:problem_description} describes the planning setting and the two-stage decision structure. Section~\ref{subsec:arc_based} formulates the problem as an arc-based two-stage stochastic program on a time-expanded network. Section~\ref{subsec:path_packing} reformulates the second-stage recourse as a path-packing problem, which is the representation used throughout the solution methodology. The arc-based formulation provides the most direct mathematical representation of the planning problem, whereas the equivalent path-packing formulation is introduced subsequently because it enables column generation and the proposed decomposition algorithm.}

\subsection{Problem Description}
\label{subsec:problem_description}

We consider a centralized UAV logistics planner operating within regulated and capacity-constrained low-altitude airspace. The planner is responsible for reserving airspace capacity prior to operations and subsequently allocating this capacity to delivery requests with specified origins and destinations. Its objective is to maximize expected profit under stochastic demand, defined as the total revenue obtained from successfully completed delivery tasks minus the cost incurred for reserving airspace capacity.

Airspace is organized as a directed corridor network over discrete time periods. 
Let $\mathcal{L}$ denote the set of airspace corridors and $\mathcal{T}$ the set of time periods.  Each corridor--time pair $(\ell,t) \in \mathcal{L} \times \mathcal{T}$ provides a limited number of flight slots that must be reserved before operations begin. These flight slots constitute corridor- and time-specific capacity units and incur reservation costs. Capacity reserved for one corridor at a given time period cannot be transferred to other corridors or other time periods. 
\REVII{The planning process follows a two-stage sequence. In the first stage, prior to demand realization, the planner reserves $x_{\ell t}$ capacity units on each corridor--time pair; because capacity is corridor- and time-specific and non-transferable, these reservations determine the capacity limits for all subsequent routing decisions. In the second stage, demand realizes. The underlying request distribution is not assumed to be fully known; we approximate it by a finite set of representative demand scenarios, each specifying a realized set of delivery requests. After observing the realized scenario, the planner selects which requests to accept and routes the accepted ones along feasible time-respecting paths within the reserved capacities.}
To illustrate the spatio-temporal nature of capacity consumption, consider the following example.

\begin{example}[Example: Spatio-temporal capacity consumption]
Consider the network shown in Figure~\ref{fig:toy_network_time}, where nodes represent take-off and landing locations and directed arcs represent air corridors. Time is discretized into periods $\{t,t+1,t+2,t+3,\ldots\}$ , and traversing any corridor requires one time period.
Suppose a delivery request is released at origin node $1$ at time $t$ and destination at node $6$. One feasible route is the path $1 \rightarrow 2 \rightarrow 5 \rightarrow 6$. If the request departs at time $t$, it occupies corridor $\ell_{12}$ during period $t$, corridor $\ell_{25}$ during period $t+1$, and corridor $\ell_{56}$ during period $t+2$.
Selecting this route therefore consumes one unit of capacity from each of the corridor–time pairs $(\ell_{12},t)$, $(\ell_{25},t+1)$, and $(\ell_{56},t+2)$.
Feasibility of this path requires that capacity be reserved for all three corridor–time units. If capacity is unavailable for any one of them, the entire route becomes infeasible. This example highlights two key structural features. First, capacity consumption is sequential and spans multiple time periods. Second, each accepted request requires jointly corridor–time resources rather than a single aggregate capacity unit. 
As a result, first-stage reservation decisions must anticipate the joint capacity requirements of feasible paths, while second-stage allocation decisions must respect the joint availability of corridor–time resources. 
$\hfill \blacksquare$
\end{example}

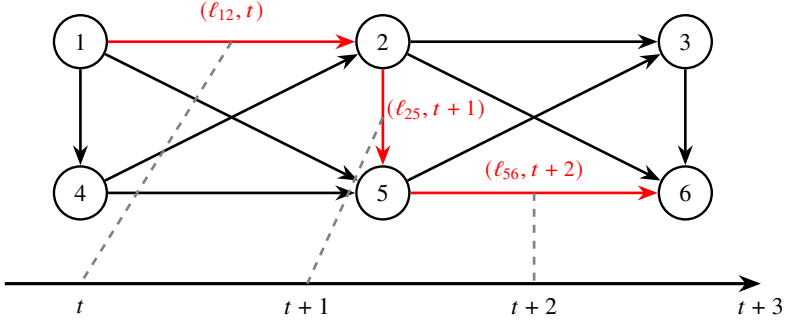
\begin{figure}[t]
\centering
\begin{tikzpicture}[
    >=Stealth,
    node/.style={circle, draw=black, line width=0.9pt, minimum size=7mm, inner sep=0pt},
    arc/.style={-Stealth, line width=1.0pt},
    pathA/.style={-Stealth, line width=1.0pt, red},
    timeline/.style={-Stealth, line width=1.2pt, black},
    proj/.style={dashed, line width=1.0pt, gray}
]

\node[node] (v1) at (0,2) {1};
\node[node] (v2) at (4,2) {2};
\node[node] (v3) at (8,2) {3};

\node[node] (v4) at (0,0) {4};
\node[node] (v5) at (4,0) {5};
\node[node] (v6) at (8,0) {6};

\draw[arc] (v2) -- (v3);
\draw[arc] (v4) -- (v5);
\draw[arc] (v1) -- (v4);
\draw[arc] (v3) -- (v6);
\draw[arc] (v1) -- (v5);
\draw[arc] (v4) -- (v2);
\draw[arc] (v2) -- (v6);
\draw[arc] (v5) -- (v3);

\draw[pathA] (v1) -- (v2);
\draw[pathA] (v2) -- (v5);
\draw[pathA] (v5) -- (v6);

\draw[timeline] (-1,-1.2) -- (9,-1.2);

\node at (0,-1.5) {$t$};
\node at (3,-1.5) {$t+1$};
\node at (6,-1.5) {$t+2$};
\node at (9,-1.5) {$t+3$};

\draw[proj] ($(v1)!0.5!(v2)$) -- (0,-1.2);

\draw[proj] ($(v2)!0.5!(v5)$) -- (3,-1.2);

\draw[proj] ($(v5)!0.5!(v6)$) -- (6,-1.2);

\node[red] at (2,2.4) {\small $(\ell_{12},t)$};
\node[red] at (4.7,1.1) {\small $(\ell_{25},t+1)$};
\node[red] at (6,0.3) {\small $(\ell_{56},t+2)$};

\end{tikzpicture}
\caption{Illustration of sequential corridor–time capacity consumption along a feasible path.}
\label{fig:toy_network_time}
\end{figure}

\subsection{Arc-based Two-Stage Stochastic Model}
\label{subsec:arc_based}

\begin{table}[htbp]
\centering
\caption{Notation for the two-stage stochastic corridor capacity model.}
\label{tab:notation_full}
\begin{tabular}{p{3.2cm} p{10cm}}
\toprule
\multicolumn{2}{l}{\textbf{Sets and Indices}} \\
\midrule
$\mathcal{V}$ & Set of physical nodes (take-off and landing sites). \\
$\mathcal{L}$ & Set of directed corridors; each $\ell=(i,j)$ connects node $i$ to node $j$. \\
$\mathcal{T}=\{1,\dots,T\}$ & Set of discrete time periods. \\
$\Omega$ & Set of demand scenarios. \\
$K^\omega$ & Set of realized delivery requests under scenario $\omega\in\Omega$. \\
$\mathcal{P}_k$ & Set of feasible time-expanded paths for request $k$. \\
$k$ & Index for requests. \\
$\ell$ & Index for corridors. \\
$t$ & Index for time periods. \\
$p$ & Index for feasible paths. \\
\midrule
\multicolumn{2}{l}{\textbf{Parameters}} \\
\midrule
$d_\ell$ & Physical length of corridor $\ell$. \\
$v$ & UAV cruising speed. \\
$\Delta$ & Length of one time period. \\
$\tau_\ell=\left\lceil d_\ell/(v\Delta)\right\rceil$ & Travel time (in periods) on corridor $\ell$. \\
$o_k, d_k$ & Origin and destination nodes of request $k$. \\
$t_k$ & Release time of request $k$. \\
$r_k$ & Revenue obtained if request $k$ is served. \\
$\pi_\omega$ & Probability of scenario $\omega$. \\
$c_{\ell t}$ & Reservation cost of one unit of capacity on corridor--time pair $(\ell,t)$. \\
$\bar x_{\ell t}$ & Maximum available capacity on corridor--time pair $(\ell,t)$. \\
$b_{\ell,t_0}^{\ell t}$ & Indicator equal to 1 if departure on corridor $\ell$ at time $t_0$ occupies $(\ell,t)$. \\
$a_{kp}^{\ell t}$ & Indicator equal to 1 if path $p\in\mathcal{P}_k$ occupies corridor--time pair $(\ell,t)$. \\
\midrule
\multicolumn{2}{l}{\textbf{Decision Variables}} \\
\midrule
$x_{\ell t}\in\mathbb Z_+$ & First-stage capacity reserved on corridor--time pair $(\ell,t)$. \\
$f_{k,\ell,t}^\omega\in\{0,1\}$ & 1 if request $k$ departs on corridor $\ell$ at time $t$ under scenario $\omega$. \\
$z_k^\omega\in\{0,1\}$ & 1 if request $k$ is served under scenario $\omega$. \\
$h_{k t}^\omega\in\{0,1\}$ & 1 if request $k$ arrives at its destination at time $t$ under scenario $\omega$. \\
$y_{kp}^\omega\in\{0,1\}$ & 1 if request $k$ is assigned to path $p\in\mathcal{P}_k$ under scenario $\omega$. \\
\bottomrule
\end{tabular}
\end{table}

We first present the original arc-based model because it more directly represents the planning problem.
We formulate the corridor–time capacity planning problem as a two-stage stochastic program on a time-expanded network. In the first stage, the planner determines corridor–time capacity reservations. After demand is realized, the second stage selects and routes delivery requests subject to the reserved capacities. Table~\ref{tab:notation_full} summarizes the notation used in the model.

Let $G=(\mathcal{V},\mathcal{L})$ denote the directed corridor network, where each corridor $\ell = (i,j) \in\mathcal{L}$ connects node $i$ to node $j$. Let $\mathcal{T}=\{1,\dots,T\}$ denote the set of discrete time periods with period length $\Delta$.
For each corridor $\ell$, let $d_\ell$ denote its physical length and $v$ the cruising speed of UAVs. The traversal time (in periods) on corridor $\ell$ is
\begin{align}\label{eq:tau_l}
\tau_\ell \;=\; \left\lceil \frac{d_\ell}{v\Delta}\right\rceil \in \mathbb Z_+ .
\end{align}
To capture temporal dynamics, each physical node $i\in \mathcal{V}$ is replicated at each time period $t\in\mathcal{T}$, forming time-stamped nodes $(i,t)$.
For each corridor  $\ell=(i,j)$ and departure time $t$, a flight arc connects $(i,t)$ to $(j, t + \tau_{\ell})$, provided $t + \tau_{\ell} \leq T$. 

\textbf{First-stage capacity decisions.}
In the first stage, the planner selects integer capacity units $x_{\ell t}$ for each corridor-time pair $(\ell,t)\in\mathcal{L}\times\mathcal{T}$.
These decisions determine the total spatio-temporal airspace capacity available during operations. Capacity is bounded by physical limits $\bar{x}_{\ell t}$ and incurs reservation cost $c_{\ell t}$ per unit. 
\REVII{A delivery request departing on corridor $\ell$ at time $t_0$ occupies the corridor--time pairs $(\ell,t_0),(\ell,t_0+1),\dots,(\ell,t_0+\tau_\ell-1)$. To represent this occupancy, define the indicator}
\begin{align}
b_{\ell,t_0}^{\ell t}=
\begin{cases}
1, & \text{if} \quad t\in\{t_0, \dots, t_0 + \tau_\ell-1\},\\
0, & \text{otherwise},
\end{cases}
\quad \forall \ell\in\mathcal{L},\ \forall t_0\in\mathcal{T},\ \forall t\in\mathcal{T}.
\label{eq:b_indicator_l}
\end{align}

\textbf{Second-stage route decisions.}
Let $\Omega$ denote a finite set of demand scenarios, each occurring with probability $\pi_\omega$.
Under scenario $\omega\in\Omega$, a set of delivery requests $K^\omega$ is realized.
Each request $k\in K^\omega$ is characterized by origin $o_k\in \mathcal{V}$, destination $d_k\in \mathcal{V}$, release time $t_k\in\mathcal{T}$, and revenue $r_k$ if served.
The planner chooses integer capacity units $x_{\ell t}$ on each corridor--time pair:
\begin{align}
\max_{\boldsymbol{x}}\quad
& \sum_{\omega\in\Omega} \pi_\omega\, Q(\boldsymbol{x};\omega)
  \;-\; \sum_{(\ell,t)\in\mathcal{L}\times\mathcal{T}} c_{\ell t}\, x_{\ell t}
\label{eq:arc_stage1_obj_l}\\
\text{s.t.}\quad
& x_{\ell t} \in \mathbb Z_+,
\qquad &&\forall (\ell,t)\in\mathcal{L}\times\mathcal{T},
\label{eq:arc_x_int_l}\\
& 0 \le x_{\ell t} \le \bar x_{\ell t},
\qquad &&\forall (\ell,t)\in\mathcal{L}\times\mathcal{T}.
\label{eq:arc_x_ub_l}
\end{align}
Given the reserved capacity $\boldsymbol{x}$ and the realized scenario $\omega$, the planner then determines which requests to accept and how to route them through the time-expanded network.
Let $f_{k,\ell,t_0}^\omega\in\{0,1\}$ indicate whether request $k$ departs on corridor $\ell$ at time $t_0$ under scenario $\omega$. 
Let $z_k^\omega\in\{0,1\}$ indicate whether request $k$ is accepted.
and let $h_{k t}^\omega\in\{0,1\}$ indicate whether request $k$ arrives at its destination $d_k$ at time $t$.
The second-stage problem is given by
\begin{align}
Q(\boldsymbol{x};\omega)=\max_{\boldsymbol{f},\boldsymbol{z}, \boldsymbol{h}}\quad
& \sum_{k\in K^\omega} r_k\, z_k^\omega
\label{eq:arc_stage2_obj_l}\\
\text{s.t.}\quad
& \sum_{k\in K^\omega}\sum_{t_0\in\mathcal{T}}
  b_{\ell,t_0}^{\ell t}\, f_{k,\ell,t_0}^\omega
\;\le\; x_{\ell t},
\qquad &&\forall (\ell,t)\in\mathcal{L}\times\mathcal{T},
\label{eq:arc_cap_l}\\
& \sum_{\ell:\,\ell=(o_k,j)} f_{k,\ell,t_k}^\omega
\;-\; \sum_{\ell:\,\ell=(i,o_k)} f_{k,\ell,t_k-\tau_\ell}^\omega
\;=\; z_k^\omega,
\qquad &&\forall k\in K^\omega,
\label{eq:arc_source_l}\\
& \sum_{\ell:\,\ell=(i,j)} f_{k,\ell,t}^\omega
\;-\; \sum_{\ell:\,\ell=(j,i)} f_{k,\ell,t-\tau_\ell}^\omega
\;=\; 0,
\qquad &&\forall k\in K^\omega,\ \forall i\in \mathcal{V}\setminus\{o_k,d_k\},\ \forall t,
\label{eq:arc_flow_l}\\
& \sum_{t\in\mathcal{T}} h_{k t}^\omega \;=\; z_k^\omega,
\qquad &&\forall k\in K^\omega,
\label{eq:arc_arrive_once_l}\\
& \sum_{\ell:\,\ell=(i,d_k)} f_{k,\ell,t-\tau_\ell}^\omega
\;-\; \sum_{\ell:\,\ell=(d_k,j)} f_{k,\ell,t}^\omega
\;=\; h_{k t}^\omega,
\qquad &&\forall k\in K^\omega,\ \forall t\in\mathcal{T},
\label{eq:arc_sink_l}\\
& f_{k,\ell,t}^\omega \in \{0,1\},
\qquad &&\forall k,\ell,t, \\
& z_k^\omega \in \{0,1\},\qquad
h_{k t}^\omega \in \{0,1\}.
\end{align}
The objective \eqref{eq:arc_stage2_obj_l} maximizes total revenue from accepted requests.
Constraint~\eqref{eq:arc_cap_l} enforces corridor--time capacity limits.
If request $k$ departs on corridor $\ell$ at time $t_0$, it occupies the corresponding corridor--time resources for $\tau_\ell$ consecutive periods, as captured by the binary indicator $b_{\ell,t_0}^{\ell t}$.
The total capacity consumption across all routed requests cannot exceed the reserved capacity $x_{\ell t}$.
Constraint~\eqref{eq:arc_source_l} ensures that a served request originates at its release node and time.
Constraint~\eqref{eq:arc_flow_l} maintains flow balance at intermediate nodes and times.
Constraint~\eqref{eq:arc_arrive_once_l} and \eqref{eq:arc_sink_l} guarantee that each accepted request arrives at its destination at exactly one time period.
Binary variables $z_k^\omega$ allow requests to be rejected, ensuring feasibility for any first-stage capacity decision.

\subsection{Time-Expanded Path-Packing Reformulation}
\label{subsec:path_packing}

While the arc-based formulation provides a transparent representation of flow conservation in the time-expanded network, its size grows with the number of time-stamped nodes and arcs. In particular, the number of binary flow variables scales with \REFTHR{$|K^{\omega}|\times|\mathcal{L}|\times|\mathcal{T}|$}, which may become computationally burdensome for fine temporal discretizations. To obtain a more compact representation and to prepare for decomposition analysis, we derive an equivalent path-based reformulation that aggregates feasible unit flows into complete origin–destination routes.

\REVII{The reformulation is built on an explicit notion of a \textit{feasible time-expanded path}. Recall from Section~\ref{subsec:arc_based} that the time-expanded network consists of \textit{flight arcs} from $(i,t)$ to $(j,t+\tau_\ell)$ for each corridor $\ell=(i,j)\in\mathcal{L}$ and each departure time $t$ with $t+\tau_\ell\le T$. Every arc strictly increases the time index, so the time-expanded network is a finite directed acyclic graph. }

\REVII{\begin{definition}[Definition 1 (Feasible time-expanded path)]
A \textit{feasible time-expanded path} for request $k\in K^\omega$ is a directed path
$p=\big((v_0,t_0),(v_1,t_1),\dots,(v_m,t_m)\big)$
in the time-expanded network such that
(i)~$(v_0,t_0)=(o_k,t_k)$, i.e., the path starts at the origin of request $k$ exactly at its release time;
(ii)~$v_m=d_k$ with $t_m\le T$, i.e., the path ends at the destination of $k$ at \textit{some} time within the planning horizon; and
(iii)~every consecutive pair of nodes is connected either by a flight arc --- the request traverses corridor $\ell=(v_s,v_{s+1})$, departing at time $t_s$ and arriving at $t_{s+1}=t_s+\tau_\ell$.
The set of all feasible time-expanded paths for request $k$ is denoted $\mathcal{P}_k$.
\end{definition}}

\REVII{A path in $\mathcal{P}_k$ represents a complete \textit{routing schedule}, including the corridor sequence and departure times. Thus, schedules following the same physical route but with different departure times are treated as distinct paths because they occupy different corridor–time units. The destination is specified as any $(d_k,t)$ since the model imposes no delivery deadline beyond the planning horizon; additional constraints such as delivery deadlines, endurance, payload, or detour limits can be incorporated by excluding infeasible paths from $\mathcal{P}_k$. Although $\mathcal{P}_k$ is finite because the time-expanded network is acyclic and finite, its size grows combinatorially with the network and planning horizon, motivating the use of column generation to generate paths on demand rather than enumerate them explicitly.}

\REVII{By construction, every $p\in\mathcal{P}_k$ satisfies the flow-conservation constraints \eqref{eq:arc_source_l}--\eqref{eq:arc_sink_l} of the arc-based formulation and represents a feasible routing schedule for request $k$; Proposition~\ref{prop:equivalence} below establishes the converse, namely that every arc-based routing of an accepted request traces out exactly one such path.}
For a path $p\in\mathcal{P}_k$,
define the corridor--time occupancy indicator
\begin{align}
a_{kp}^{\ell t}
\;:=\;\sum_{t_0\in\mathcal{T}:\ p\ \text{departs on}\ \ell\ \text{at}\ t_0} b_{\ell,t_0}^{\ell t}
\;=\;
\begin{cases}
1 & \text{if path $p$ occupies corridor--time pair $(\ell,t)$},\\
0 & \text{otherwise}.
\end{cases}
\label{eq:a_indicator}
\end{align}
\REVII{The first expression writes $a_{kp}^{\ell t}$ in terms of the occupancy indicator $b_{\ell,t_0}^{\ell t}$ of the arc-based model~\eqref{eq:b_indicator_l}: if $p$ departs on corridor $\ell$ at time $t_0$, then $a_{kp}^{\ell t}=1$ for every $t\in\{t_0,\dots,t_0+\tau_\ell-1\}$, so the multi-period traversal of each corridor is aggregated into the occupancy footprint of the path. The indicator therefore exactly represents the capacity consumption implied by the arc-based model.}
Let $y_{kp}^\omega\in\{0,1\}$ indicate whether request $k$
is assigned to path $p$ under scenario $\omega$.
The second-stage problem can then be written as the following
capacitated path-packing problem:
\begin{align}
G(\boldsymbol{x};\omega)=\max_{\boldsymbol{y}}\quad
& \sum_{k\in K^\omega} \sum_{p\in \mathcal{P}_k} r_k\, y_{kp}^\omega
\label{eq:stage2_obj_path}\\
\text{s.t.}\quad
& \sum_{p\in \mathcal{P}_k} y_{kp}^\omega \le 1,
\qquad &&\forall k\in K^\omega,
\label{eq:link_eq}\\
& \sum_{k\in K^\omega} \sum_{p\in \mathcal{P}_k}
  a_{kp}^{\ell t} y_{kp}^\omega
  \le x_{\ell t},
\qquad &&\forall (\ell,t)\in\mathcal{L}\times\mathcal{T},
\label{eq:cap}\\
& y_{kp}^\omega \in \{0,1\},
\qquad &&\forall k\in K^\omega,\ p\in\mathcal{P}_k.
\end{align}
Constraint~\eqref{eq:link_eq} ensures that each request is assigned to at most one feasible path. Constraint~\eqref{eq:cap} enforces the same corridor--time capacity limits as in the arc-based formulation. Because requests may remain unassigned, the recourse problem is feasible for any first-stage decision $x$.

\REFTHR{The path-based formulation is equivalent to the arc-based model: the first-stage problem \eqref{eq:arc_stage1_obj_l}--\eqref{eq:arc_x_ub_l} is unchanged, and the reformulation preserves the second-stage feasible set and optimal value while eliminating the explicit flow-conservation constraints. We formalize this equivalence in Proposition~\ref{prop:equivalence} and Corollary~\ref{cor:lp_equivalence}.}

\REFTHR{\begin{proposition}[Equivalence of the arc-based and path-packing recourse]\label{prop:equivalence}
Fix a scenario $\omega\in\Omega$ and a first-stage capacity vector $\boldsymbol{x}\in\mathbb{Z}_+^{|\mathcal{L}|\times|\mathcal{T}|}$. Let the time-expanded network be constructed as in Section~\ref{subsec:arc_based}, and let $\mathcal{P}_k$ be the full path set of Definition~1. Denote by $\mathcal{F}_A(\boldsymbol{x};\omega)$ the feasible set of the arc-based recourse \eqref{eq:arc_stage2_obj_l}--\eqref{eq:arc_sink_l}, including the binary requirements, and by $\mathcal{F}_P(\boldsymbol{x};\omega)$ that of the path-packing recourse \eqref{eq:stage2_obj_path}--\eqref{eq:cap}. Then there exists an injective, objective-value-preserving mapping $\Psi:\mathcal{F}_P(\boldsymbol{x};\omega)\to\mathcal{F}_A(\boldsymbol{x};\omega)$. 
In particular, $Q(\boldsymbol{x};\omega)=G(\boldsymbol{x};\omega)$, and every optimal solution of the path-packing recourse corresponds, via $\Psi$, to an optimal solution of the arc-based recourse.
\end{proposition}}

\REFTHR{The proof, given in Appendix~\ref{app:proofs}, establishes the equivalence by decomposing flows in the acyclic time-expanded network.} \REVII{While Proposition~\ref{prop:equivalence} considers integral solutions, the following corollary extends the equivalence to the LP relaxation used in the Benders decomposition, including the corresponding dual variables.}

\REVII{\begin{corollary}[LP equivalence and dual correspondence]\label{cor:lp_equivalence}
Fix $\omega\in\Omega$. For every $\boldsymbol{x}\ge 0$, not necessarily integer, the LP relaxations of the two recourse formulations coincide in value: $G_{\mathrm{LP}}^{\mathrm{AF}}(\boldsymbol{x};\omega)=G_{\mathrm{LP}}(\boldsymbol{x};\omega)$, where $G_{\mathrm{LP}}^{\mathrm{AF}}$ denotes the LP relaxation of \eqref{eq:arc_stage2_obj_l}--\eqref{eq:arc_sink_l}, stated explicitly as the arc-flow subproblem in Section~\ref{sec:subproblem_arcflow}, and $G_{\mathrm{LP}}$ that of \eqref{eq:stage2_obj_path}--\eqref{eq:cap}. Moreover, the two LPs have identical sets of optimal dual multipliers associated with the capacity constraints \eqref{eq:arc_cap_l} and \eqref{eq:cap}.
\end{corollary}}

\REVII{Proposition~\ref{prop:equivalence} and Corollary~\ref{cor:lp_equivalence} show three points.
First, Corollary~\ref{cor:lp_equivalence} makes the Benders optimality cut of Section~\ref{sec:decomposition} solver-independent: any exact LP solver for the scenario recourse --- in this paper, column generation on the path formulation (Section~\ref{sec:subproblem_cg}) --- produces the same supporting hyperplane of the recourse value function, so the cut's validity does not depend on which formulation or algorithm computes it.
Second, Proposition~\ref{prop:equivalence} emphasizes the recovery of deployable integer plans, which reads acceptance and routing decisions off the path formulation once the first stage is fixed.
Third, the equivalence is structural and does not imply that the recourse LP is integral at a fixed integer  $\boldsymbol{x}$, the latter is a separate property, verified numerically
in in Appendix~\ref{sec:app_integrality}.}

\medskip
\section{Solution Methodology}
\label{sec:analysis_algorithm}

\REVII{
This section presents the solution framework for the proposed two-stage stochastic capacity planning model.
The framework exploits the separable structure of the recourse problem: for a fixed first-stage capacity decision, each scenario reduces to an independent LP-relaxed path-packing problem, whose dual multipliers provide valid supporting hyperplanes of the recourse function.
We therefore develop a Benders decomposition framework (Section~\ref{sec:decomposition}, Algorithm~\ref{alg:multicut_benders}) with scenario-wise optimality cuts.
Each Benders iteration requires the exact evaluation of the LP-relaxed scenario recourse and the associated capacity dual multipliers, which are obtained through a column-generation procedure over the full path space (Section~\ref{sec:subproblem_cg}). The truncated $P$-shortest path sets constructed in Section~\ref{sec:path_generation} are used only as warm-start column pools and do not restrict the final recourse evaluation.
Because the Benders procedure operates on the LP relaxation, an integer recovery procedure is subsequently applied to obtain implementable first-stage capacity decisions and second-stage routing decisions (Section~\ref{sec:integrality}).}

\subsection{Benders Decomposition}
\label{sec:decomposition}
We solve the two-stage stochastic model using a \textit{disaggregated Benders decomposition with column generation} (\DBCG).
The decomposition is applicable because the first-stage capacity vector $x$ enters each scenario recourse problem only through the corridor--time capacity constraints. For a fixed $x$, each scenario subproblem reduces to an independent linear path-packing problem over the time-expanded network. Therefore, each scenario recourse function can be approximated through scenario-specific dual information from the subproblems. The disaggregated structure preserves individual scenario cuts rather than aggregating them into a single expected cut. This choice is motivated by the fact that different demand realizations induce different congestion patterns and therefore distinct marginal values of corridor--time capacity.

\REFTWO{To enable Benders decomposition, we relax the second-stage path-selection variables $y_{kp}^\omega$ to $[0,1]$ in each scenario subproblem and allow the first-stage capacity variables $x_{\ell t}$ to take continuous values in the master problem. This relaxation yields a linear recourse problem whose dual multipliers associated with the corridor--time capacity constraints can be used to construct Benders optimality cuts. Note that the path-packing recourse formulation does not explicitly include the acceptance variable $z_k^\omega$; instead, demand acceptance is implicitly represented by the assignment constraint $\sum_{p\in\mathcal{P}_k} y_{kp}^\omega\leq 1$. Therefore, relaxing $y_{kp}^\omega$ jointly relaxes the routing and acceptance decisions, and no additional relaxation of $z_k^\omega$ is required.}

\textbf{Master problem}. To reformulate the problem, we introduce auxiliary variables $\theta_\omega$ to represent the LP-relaxed recourse value under scenario $\omega$.
Let $\mathcal{C}_\omega$ denote the set of Benders cuts accumulated for scenario $\omega$.
At iteration $n$, the master problem takes the form
\begin{align}
\text{(Master)}\qquad \max_{\boldsymbol{x},\theta_\omega} \quad
& \sum_{\omega\in\Omega} \pi_\omega \theta_\omega
  - \sum_{(\ell,t)\in \mathcal{L}\times\mathcal{T}} c_{\ell t} x_{\ell t}
\label{eq:master_obj}\\
\text{s.t.}\quad
& 0 \leq x_{\ell t} \leq \bar{x}_{\ell t},
\qquad &&\forall (\ell,t)\in \mathcal{L}\times\mathcal{T},
\label{eq:master_x_box}\\
&\theta_\omega \;\le\; \sum_{k\in K^\omega} r_k,
\qquad &&\forall \omega\in\Omega,
\label{eq:master_theta_ub}\\
& \theta_\omega \leq (\lambda^{\omega j})^\top x + \beta^{\omega j},
\qquad &&\forall \omega\in\Omega,\; j\in\mathcal{C}_\omega.
\label{eq:master_cuts}
\end{align}
The master problem selects a capacity vector $\boldsymbol{x}$ and enforces previously generated supporting hyperplanes that upper-bound the recourse function.
\REVII{Constraint~\eqref{eq:master_theta_ub} is the trivial scenario-wise bound $G_{\mathrm{LP}}(x;\omega)\le \sum_{k\in K^\omega} r_k$ (serve every request); it keeps the master bounded at iteration $n=0$, when $\mathcal{C}_\omega=\emptyset$, and is dominated by every subsequently generated cut, so the LP-relaxed optimum is unchanged.}

\textbf{Scenario subproblem}.
Given a candidate first-stage solution $\boldsymbol{x}^n$ and a scenario $\omega$, the LP-relaxed recourse problem is the capacitated path-packing linear program
\begin{align}
\text{(Sub)}\qquad 
G_{\mathrm{LP}}(\boldsymbol{x}^n;\omega)=\max_{\boldsymbol{y}^\omega}\quad
& \sum_{k\in K^\omega} \sum_{p\in\mathcal{P}_k} r_k\, y_{kp}^\omega
\label{eq:sub_obj}\\
\text{s.t.}\quad
& \sum_{p\in\mathcal{P}_k} y_{kp}^\omega \le 1,
\qquad &&\forall k\in K^\omega,
\label{eq:sub_assign}\\
& \sum_{k\in K^\omega}\sum_{p\in\mathcal{P}_k} a_{kp}^{\ell t} y_{kp}^\omega \le x^n_{\ell t},
\qquad &&\forall (\ell,t)\in \mathcal{L}\times\mathcal{T},
\label{eq:sub_cap}\\
& 0 \le y_{kp}^\omega \le 1,
\qquad &&\forall k\in K^\omega,\; p\in\mathcal{P}_k.
\label{eq:sub_box}
\end{align}

\REFTWO{\textbf{No feasibility cuts are required.}
The path-packing recourse~\eqref{eq:sub_obj}--\eqref{eq:sub_box} is feasible for every $x\ge 0$ because the trivial assignment $y_{kp}^\omega\equiv 0$ (i.e., rejecting all requests) satisfies~\eqref{eq:sub_assign}--\eqref{eq:sub_cap} with slack and is feasible. Equivalently, the recourse value satisfies $0\le G_{\mathrm{LP}}(x;\omega)\le \sum_{k\in K^\omega} r_k$ for every $x$, so the second-stage dual is always bounded and Algorithm~\ref{alg:multicut_benders} never needs to generate feasibility (extreme-ray) cuts; only optimality cuts~\eqref{eq:opt_cut} are added. This property follows from the explicit rejection option in the recourse model and is a structural feature rather than a numerical artifact.
The same boundedness property also motivates the initialization constraint in the master problem.}

\textbf{Scenario-wise optimality cut.}
For each scenario $\omega$, solving the subproblem at $\boldsymbol{x}^n$ yields both the primal value $G_{\mathrm{LP}}(\boldsymbol{x}^n;\omega)$ and an optimal dual multiplier $\lambda^{\omega *}(\boldsymbol{x}^n)$ associated with the capacity constraints~\eqref{eq:sub_cap}.
Define the intercept of the supporting hyperplane by
\begin{align}
\beta^{\omega *}(\boldsymbol{x}^n)
:=G_{\mathrm{LP}}(\boldsymbol{x}^n;\omega)-
\big(\lambda^{\omega *}(\boldsymbol{x}^n)\big)^\top \boldsymbol{x}^n .
\label{eq:beta_def}
\end{align}
\REVII{Because the recourse problem is a maximization problem, its dual representation provides an upper approximation of the recourse function. By weak LP duality, the following scenario-wise optimality cut is valid:}
\begin{align}
\theta_\omega \le \big(\lambda^{\omega *}(\boldsymbol{x}^n)\big)^\top x + \beta^{\omega *}(\boldsymbol{x}^n),
\qquad \forall \omega\in\Omega.
\label{eq:opt_cut}
\end{align}
\REVII{The validity of~\eqref{eq:opt_cut} follows from weak LP duality of the scenario subproblem; the full derivation is given in Appendix~\ref{app:proofs}.}
\REVII{Because each cut is valid for every $\boldsymbol{x}$, repeated addition of the cuts \eqref{eq:opt_cut} yields a progressively tighter polyhedral outer approximation of the recourse function.
In \DBCG, one scenario-specific optimality cut is generated for each scenario at every iteration. This preserves heterogeneous congestion information across demand realizations.} 
\REVII{This procedure constitutes the proposed \DBCG algorithm and is summarized
in Algorithm~\ref{alg:multicut_benders}. The master objective provides an upper bound on the LP-relaxed problem, while evaluating the recourse at the current solution provides a valid lower bound. The algorithm terminates when the gap between these bounds is within a prescribed tolerance $\varepsilon$.}

\begin{algorithm}[!h]
\caption{\DBCG for the LP-relaxed Model}
\label{alg:multicut_benders}
\KwIn{Scenario set $\Omega$ with probabilities $\{\pi_\omega\}$; feasible set $\boldsymbol{X}$; tolerance $\varepsilon$; iteration limit $I_{\max}$; time limit $T_{\max}$.}
\KwOut{Incumbent capacity plan $\boldsymbol{x}^*$.}

Initialize cut sets $\mathcal{C}_\omega \leftarrow \emptyset$ for all $\omega\in\Omega$\;
\rev{Impose trivial upper bound $\theta_\omega \le \sum_{k\in K^\omega} r_k$ for all $\omega\in\Omega$ so that the master is well-posed at $n=0$\;}
Set $\mathrm{UB}\leftarrow +\infty$, $\mathrm{LB}\leftarrow -\infty$, and iteration counter $n\leftarrow 0$\;
Initialize incumbent solution $\boldsymbol{x}^* \in X$ arbitrarily\;

\While{$n<I_{\max}$ \textbf{and} \texttt{runtime}$<T_{\max}$ \textbf{and} $\mathrm{UB}-\mathrm{LB}>\varepsilon$}{
Solve the master problem to obtain $(\boldsymbol{x}^{n},\{\theta_\omega^{n}\}_{\omega\in\Omega})$\;
Set $\mathrm{UB}\leftarrow \sum_{\omega\in\Omega} \pi_\omega \theta_\omega^{n}-\sum_{(\ell,t)\in\mathcal{L}\times\mathcal{T}} c_{\ell t}x_{\ell t}^{n}$\;

\ForEach{$\omega\in\Omega$}{
\REVII{Solve the scenario subproblem at $\boldsymbol{x}^{n}$ by column generation (Algorithm~\ref{alg:cg_subproblem}, with the persistent pool $\{\widetilde{\mathcal{P}}_k\}$) to obtain $G_{\mathrm{LP}}(\boldsymbol{x}^{n};\omega)$ and the dual multiplier $\lambda^{\omega *}(\boldsymbol{x}^{n})$ associated with~\eqref{eq:sub_cap}\;}
Compute $\beta^{\omega *}(\boldsymbol{x}^{n}) \leftarrow G_{\mathrm{LP}}(\boldsymbol{x}^{n};\omega)-\big(\lambda^{\omega *}(\boldsymbol{x}^{n})\big)^\top \boldsymbol{x}^{n}$\;
Add the cut $\theta_\omega \leq \big(\lambda^{\omega *}(\boldsymbol{x}^{n})\big)^\top x + \beta^{\omega *}(\boldsymbol{x}^{n})$ to $\mathcal{C}_\omega$\;
}

Compute the candidate lower bound
$\mathrm{LB}^{n} \leftarrow \sum_{\omega\in\Omega} \pi_\omega G_{\mathrm{LP}}(\boldsymbol{x}^{n};\omega)-\sum_{(\ell,t)\in\mathcal{L}\times\mathcal{T}} c_{\ell t}\boldsymbol{x}_{\ell t}^{n}$\;

\If{$\mathrm{LB}^{n}>\mathrm{LB}$}{
Update $\mathrm{LB}\leftarrow \mathrm{LB}^{n}$ and set $\boldsymbol{x}^*\leftarrow \boldsymbol{x}^{n}$\;
}

Set $n\leftarrow n+1$\;}
\Return{$\boldsymbol{x}^*$}\;
\end{algorithm}

\subsection{Scenario Subproblem Solution}
\label{sec:subproblem}
The Benders framework requires repeated evaluation of scenario recourse problems. We consider two recourse solution strategies depending on the desired accuracy and computational scale. For the proposed \DBCG algorithm, each scenario subproblem is solved exactly through column generation, where a dynamically generated path set recovers the full LP-relaxed recourse value without explicit path enumeration. For larger networks, we introduce disaggregated Benders with truncated-path recourse (\DBTP), which replaces exact column generation with a truncated path representation. By restricting each request to a fixed candidate path pool, \DBTP avoids repeated pricing operations and provides a scalable approximation of the same recourse structure.

\rev{\subsubsection{Exact Evaluation of LP Recourse via Column Generation}
\label{sec:subproblem_cg}
For \DBCG, the scenario subproblems are solved through column generation. The complete path set $\mathcal P_k$ is potentially exponential, making explicit enumeration impractical. Column generation avoids this issue by maintaining a restricted path pool and iteratively introducing paths with positive reduced cost.
Let $\widehat G_{\mathrm{LP}}(\boldsymbol{x}^n;\omega)$ denote the recourse value obtained when each request is restricted to its truncated path pool $\widehat{\mathcal P}_k$. If an omitted path has positive reduced cost, then $\widehat G_{\mathrm{LP}}(\boldsymbol{x}^n;\omega)<G_{\mathrm{LP}}(\boldsymbol{x}^n;\omega)$, and the resulting Benders cut may not accurately represent the full recourse function. To eliminate this truncation error without explicitly enumerating the complete path set $\mathcal P_k$, we solve each scenario subproblem~\eqref{eq:sub_obj}--\eqref{eq:sub_box} by column generation, using $\widehat{\mathcal P}_k$ as a warm-start column pool.

Fix scenario $\omega$ and capacity vector $\boldsymbol{x}^n$. Let $\widetilde{\mathcal{P}}_k\subseteq\mathcal{P}_k$ denote the current restricted column pool for request $k$, initialized as $\widetilde{\mathcal{P}}_k\leftarrow\widehat{\mathcal{P}}_k$ from Algorithm~\ref{alg:p_short}. 
The \textit{restricted master LP} is the scenario subproblem \eqref{eq:sub_obj}--\eqref{eq:sub_box} with each $\mathcal{P}_k$ replaced by the pool $\widetilde{\mathcal{P}}_k$; denote its optimal value by $\widetilde G_{\mathrm{LP}}(\boldsymbol{x}^n;\omega)\le G_{\mathrm{LP}}(\boldsymbol{x}^n;\omega)$. Column generation is performed only within each Benders subproblem and does not alter the outer decomposition structure.
Solving the restricted master LP yields dual multipliers $\alpha_k^\omega\ge 0$ for the assignment constraints and $\lambda_{\ell t}^\omega\ge 0$ for the capacity constraints. The pricing problem for request $k$ identifies a path $p^*\in\mathcal{P}_k$ of maximum reduced cost
\begin{align}
\bar r_k(p)\;=\;r_k\;-\;\alpha_k^\omega\;-\;\sum_{(\ell,t)\in p}\lambda_{\ell t}^\omega,
\label{eq:reduced_cost}
\end{align}
where the summation runs over the corridor--time pairs occupied by $p$. Maximizing $\bar r_k(p)$ over $\mathcal{P}_k$ is equivalent to a longest-reduced-cost path problem on the time-expanded directed acyclic graph from $(o_k,t_k)$ to any $(d_k,t)$ with $t\le T$, where each flight arc on corridor $\ell$ departing at time $t$ has cost $\sum_{s=t}^{t+\tau_\ell-1}\lambda_{\ell s}^\omega$. Since $\lambda_{\ell t}^\omega\ge 0$, every additional flight arc weakly decreases the reduced cost, so the maximizing path never benefits from revisiting a physical node except when the arcs involved carry zero dual price; the pricing DP does not explicitly exclude such degenerate revisits, but they carry no reduced-cost advantage over the direct route.
Because the time index is monotonic, the network is acyclic and each pricing iteration is solved in $O(|\mathcal{V}|T+|\mathcal{L}|T)$ per request by a single topological sweep in the current time-expanded network. If the optimal pricing value satisfies $\bar r_k(p^*)>\varepsilon_{\mathrm{cg}}$ for some $k$, the path $p^*$ is appended to $\widetilde{\mathcal{P}}_k$ and the restricted master is re-solved; otherwise the dual multipliers are optimal for the unrestricted LP and $\widetilde G_{\mathrm{LP}}(\boldsymbol{x}^n;\omega)=G_{\mathrm{LP}}(\boldsymbol{x}^n;\omega)$. 
The Benders cut~\eqref{eq:opt_cut} is then formed with the CG-optimal $(\lambda^{\omega *},G_{\mathrm{LP}}(\boldsymbol{x}^n;\omega))$ and is therefore tight relative to the full LP recourse. 
At termination, no request admits a positive reduced-cost path, implying that the restricted master solution is optimal for the full path-based LP recourse. The resulting dual multipliers are therefore valid for constructing the \DBCG Benders optimality cuts.
The procedure is summarized in Algorithm~\ref{alg:cg_subproblem}.

\begin{algorithm}[!h]
\caption{Column Generation Algorithm for Exact Evaluation of LP Scenario Recourse}
\label{alg:cg_subproblem}
\KwIn{Scenario $\omega$ with request set $K^\omega$; capacity vector $\boldsymbol{x}^n$; warm-start pools $\{\widehat{\mathcal{P}}_k\}$ from Algorithm~\ref{alg:p_short}; pricing tolerance $\varepsilon_{\mathrm{cg}}>0$; CG iteration limit $J_{\max}$.}
\KwOut{$G_{\mathrm{LP}}(\boldsymbol{x}^n;\omega)$, dual multipliers $(\alpha^{\omega *},\lambda^{\omega *})$, and updated pool $\{\widetilde{\mathcal{P}}_k\}$.}

Initialize $\widetilde{\mathcal{P}}_k\leftarrow \widehat{\mathcal{P}}_k$ for all $k\in K^\omega$\;

\For{$j=0$ \KwTo $J_{\max}-1$}{
  Solve the restricted master LP on $\{\widetilde{\mathcal{P}}_k\}_{k\in K^\omega}$ to obtain primal $\boldsymbol{y}^{\omega(j)}$ and duals $(\alpha^{\omega(j)},\lambda^{\omega(j)})$\;
  Set $p^*\leftarrow \text{None}$, $\bar r^*\leftarrow \varepsilon_{\mathrm{cg}}$\;
  \ForEach{$k\in K^\omega$}{
    Solve the pricing problem: find $p_k^{\dagger}\in\mathcal{P}_k$ maximizing~\eqref{eq:reduced_cost} via a single forward dynamic-programming pass over the topologically ordered time-expanded DAG. Flight arcs on corridor $\ell$ departing at time $t$ contribute $-\sum_{s=t}^{t+\tau_\ell-1}\lambda_{\ell s}^{\omega(j)}$ to the reduced cost\;
    \If{$\bar r_k(p_k^{\dagger})>\bar r^*$}{
      Update $p^*\leftarrow p_k^{\dagger}$, $k^*\leftarrow k$, $\bar r^*\leftarrow \bar r_k(p_k^{\dagger})$\;
    }
  }
  \If{$\bar r^*\leq\varepsilon_{\mathrm{cg}}$}{
    Set $G_{\mathrm{LP}}(\boldsymbol{x}^n;\omega)\leftarrow \widetilde G_{\mathrm{LP}}(\boldsymbol{x}^n;\omega)$, $(\alpha^{\omega *},\lambda^{\omega *})\leftarrow (\alpha^{\omega(j)},\lambda^{\omega(j)})$, and \textbf{break}\;
  }
  Append $p^*$ to $\widetilde{\mathcal{P}}_{k^*}$\;
}
\Return{$G_{\mathrm{LP}}(\boldsymbol{x}^n;\omega)$, $(\alpha^{\omega *},\lambda^{\omega *})$, $\{\widetilde{\mathcal{P}}_k\}$}\;
\end{algorithm}

Two implementation choices are adopted. First, the column pool is maintained across Benders iterations, allowing each scenario subproblem to be warm-started from previously generated columns. Second, the initial path pool generated by Algorithm~\ref{alg:p_short} is independent of scenarios and Benders iterates, and therefore can be constructed once for each request.}

\subsubsection{Approximate Recourse Evaluation via Truncated Path Pools}
\label{sec:path_generation}

For \DBTP, the exact pricing procedure is replaced by a fixed candidate path representation. The objective is not to recover the full recourse value, but to obtain a computationally efficient approximation that preserves the main routing alternatives.
Therefore, DisBen-TP and DisBen-CG solve the same Benders formulation; they differ only in the initialization strategy of the restricted master problems.

The truncated-path recourse requires a finite candidate set $\mathcal{P}_k$ for each request $k$.
We therefore construct restricted candidate sets using a truncated $P$-shortest-path procedure.
For each origin--destination pair $(o_k,d_k)$, we first compute the shortest path with respect to physical travel time.
We then generate up to $P$ shortest simple paths using $k$-shortest path algorithms \citep{yen1971finding, chen2020efficient}. 
\REVII{The deviation-based enumeration procedure which adapted from Yen's $k$-shortest-path algorithm is summarized in Algorithm~\ref{alg:p_short} in Appendix~\ref{app:pathgen}, together with a detailed account of the deviation step.}
To avoid excessively circuitous routes, any path whose total travel time exceeds $\rho$ times the shortest-path travel time is discarded, where $\rho \geq 1$ is a prescribed detour ratio.
Consequently, each request satisfies $|\mathcal{P}_k| \leq P$. This detour and count restriction applies only to the fixed candidate pool $\widehat{\mathcal P}_k$ used to warm-start \DBCG and, unmodified, as the final candidate set for \DBTP; it does not constrain the exact pricing DP of Algorithm~\ref{alg:cg_subproblem}, which searches unrestricted over the full time-expanded DAG and therefore recovers paths of arbitrary detour whenever they carry positive reduced cost.

Time feasibility is incorporated explicitly.
Given release time $t_k$ and discrete traversal durations $\tau_\ell$ on each corridor $\ell$, we verify that the induced arrival times remain within the planning horizon $\mathcal{T}$.
Paths violating temporal feasibility are removed.
To further reduce redundancy, dominated paths are eliminated.
A path $p\in\mathcal{P}_k$ is removed if there exists another path $p'\in\mathcal{P}_k$ such that:
(i) the total travel time of $p'$ is no greater than that of $p$, and
(ii) for every corridor--time pair $(\ell,t)$, the resource usage of $p'$ weakly dominates that of $p$.
This pruning step removes strictly inferior routing alternatives while preserving meaningful substitution flexibility.

The resulting candidate sets $\mathcal{P}_k$ maintain tractability of the recourse problem while retaining sufficient routing diversity to capture congestion-induced trade-offs across corridor--time resources.

\REFTWO{\subsection{Integer Solution Recovery and Evaluation}\label{sec:integrality}}

\REFTWO{The Benders procedure described above solves the LP relaxation of the two-stage model.
To obtain implementable decisions, we perform a post-processing recovery procedure that converts the fractional capacity plan into an integer reservation plan and subsequently solves the integer recourse problem.
The resulting solution provides a feasible lower bound for the original stochastic program, while the LP-relaxed Benders objective remains a valid upper bound.}

\REVII{\textit{Step~1 (First-stage rounding).}
Let $\boldsymbol{x}_{\mathrm{LP}}^*$ be the optimal capacity vector returned by Algorithm~\ref{alg:multicut_benders} on the LP relaxation of the master, and let
$\mathcal{S}=\{(\ell,t):x^{*}_{\mathrm{LP},\ell t}\notin\mathbb{Z}_{+}\}$ denote the set of fractional slots. We form two integer candidates and retain the better of the two.
The ceiling candidate is $x^{\uparrow}_{\ell t}=\lceil x^{*}_{\mathrm{LP},\ell t}\rceil$. It is feasible without truncation, since $\bar{x}_{\ell t}$ is integer and
$x^{*}_{\mathrm{LP},\ell t}\le\bar{x}_{\ell t}$ together imply
$x^{\uparrow}_{\ell t}\le\bar{x}_{\ell t}$, and its reservation cost exceeds that of the LP plan by exactly
$\Delta:=\sum_{(\ell,t)\in\mathcal{S}}c_{\ell t}\bigl(\lceil x^{*}_{\mathrm{LP},\ell t}\rceil-x^{*}_{\mathrm{LP},\ell t}\bigr)$.
The floor candidate is $x^{\downarrow}_{\ell t}=\lfloor x^{*}_{\mathrm{LP},\ell t}\rfloor$. It is cheaper and remains feasible for the recourse, since rejecting every request is always
admissible (Section~\ref{sec:decomposition}); flooring degrades the recourse value rather than violating feasibility. Both candidates are evaluated under the integer recourse of Step~2 across all scenarios, and the one with the larger expected profit is retained as $\boldsymbol{x}^{*}$.}

\REFTWO{\textit{Step~2 (Second-stage integer recovery).}
With $\boldsymbol{x}^*$ fixed, the realized scenario subproblem becomes the binary path-packing problem~\eqref{eq:stage2_obj_path}--\eqref{eq:cap} with $y^{\omega}_{kp}\in\{0,1\}$. We solve this binary problem directly with Gurobi on the full request set; the optional pre-screening~\eqref{eq:prescreen} is not applied here, so the recovered acceptance and routing decisions are optimal for the given capacity plan. Integrality of these decisions is guaranteed by construction and does not rely on any structural property of the constraint matrix, which, as discussed in Section~\ref{sec:decomposition}, is not totally unimodular in general. Appendix~\ref{sec:app_integrality} nonetheless reports that at the rounded plans $\lceil\boldsymbol{x}^{*}_{\mathrm{LP}}\rceil$ the LP relaxation of this binary problem attains the integer optimal value to within $10^{-6}$ on every instance tested, and is naturally integral in $85.3\%$ of cases, the remainder reflecting exact ties between routes of equal objective value. This observation explains the low computational cost of Step~2 rather than being required for its correctness. It is an empirical finding on instances with $|\mathcal{V}|\le 12$ and we do not claim that it extends to larger networks.}

\REFTWO{Together, Steps~1--2 deliver a deployable, fully integer capacity reservation plan and a fully integer set of accepted requests. The LP relaxation in Algorithm~\ref{alg:multicut_benders} is best understood as a computational device for generating valid optimality cuts; the planning decisions reported throughout the case study (Section~\ref{sec:case study}) are the integer-recovered
$(\boldsymbol{x}^*,\boldsymbol{z}^{\omega*},\boldsymbol{y}^{\omega*})$ obtained from Steps~1--2. Exactness refers only to the scenario-wise LP recourse evaluation: the outer LP-relaxed Benders loop terminates with the residual $\mathrm{gap}_{\mathrm{Bend}}$ reported
in Section~\ref{sec:benchmarks} under a fixed computational budget, and the first-stage rounding of Step~1 is bounded only in the conditional sense given in the remark above.}

\medskip
\section{Computational Enhancements}
\label{sec:acceleration}
\REVII{Although the proposed decomposition exploits scenario separability, solving realistic large-scale corridor--time planning problems remains computationally challenging. The difficulty stems from two sources: the strong degeneracy of the path-packing recourse, which weakens Benders cuts and slows convergence, and the rapid growth of candidate requests and routing alternatives as network size increases. This section therefore introduces two acceleration strategies used throughout the computational study: outer-loop stabilization (Section~\ref{sec:stabilization}) and an optional profitability-based request pre-screening rule (Section~\ref{sec:prescreening}). The first improves convergence by strengthening Benders cuts and stabilizing the master iterations, whereas the second reduces problem size before optimization. Additional implementation details, parameter sensitivity, and the numerical verification of second-stage integrality are collected in Appendix~\ref{app:diagnostics}.}

\subsection{Stabilization of the Benders Decomposition}
\label{sec:stabilization}

The path-packing recourse LP is strongly dual-degenerate in the plain Benders implementation.
In a diagnostic run without stabilization or cut strengthening, over $90\%$ of pooled columns across the benchmark instances have reduced cost within $10^{-7}$ of zero.
Consequently, the dual $(\lambda^{\omega*},\beta^{\omega*})$ returned at the master point is only one of many optima, the resulting cut is typically weak, and the master iterate may oscillate.
Moreover, if a cut is generated from an incompletely priced restricted master, it may violate dual feasibility and therefore fail to produce a valid supporting hyperplane for the concave recourse function.
These observations motivate the following acceleration strategies, which improve cut quality and maintain cut validity.

In practice, this degeneracy slows convergence because successive master iterations provide limited information about the marginal value of corridor--time capacity. We therefore introduce two complementary mechanisms: stronger cuts to improve the approximation quality and stabilization of the master trajectory.

\REVII{\textbf{Papadakos independent-point cuts.}
Dual degeneracy implies that cuts generated only at the current master solution may provide limited information about the global shape of the recourse function. To enrich the cut pool, we adopt the independent-point scheme of \citep{papadakos2008practical}.
Any dual-feasible multiplier pair yields a valid inequality
$\theta_\omega\le\lambda^\top x+\beta$ for all $x$.
At each iteration, in addition to solving the scenario subproblem at the current master solution $x^n$, we solve it at an auxiliary capacity vector $x^0$ and append the corresponding cuts.
The vector $x^0$ is initialized in the relative interior of the capacity box and updated as
$x^0\leftarrow\mu x^0+(1-\mu)x^n$ with $\mu=0.5$,
allowing the auxiliary point to gradually approach the congested operating region which may generate more informative capacity duals in congested regions.
Although related in spirit to Pareto-improving Benders cuts, the proposed scheme does not enforce the strict Magnanti--Wong auxiliary optimization.
Nevertheless, the generated cuts remain valid Benders inequalities and empirically provide tighter approximations than cuts generated only at $x^n$ (Appendix~\ref{app:diagnostics}).}

\REVII{\textbf{Dual-feasible cut repair.}
When column generation terminates due to the iteration limit rather than full price-out, the restricted-master dual solution may not be dual-feasible for the full recourse problem because a path with positive reduced cost can remain outside the restricted pool.
Using such multipliers directly would produce an invalid Benders cut that may underestimate the concave recourse function.
We restore dual feasibility by repairing the assignment dual variables:
$\alpha_k^{\mathrm{rep}}=\alpha_k+\max\{0,\bar r_k(\lambda)\}$,
where $\bar r_k(\lambda)=r_k-\alpha_k-\min_{p\in\mathcal P_k} \sum_{(\ell,t)\in p}\lambda_{\ell t}$
is the maximum reduced-cost violation identified by one additional pricing pass.
The only violated dual constraints are the path-specific assignment constraints, because the capacity dual variables already satisfy their domain restrictions.
The repaired intercept is then
$\sum_k\alpha_k^{\mathrm{rep}}$,
and weak LP duality guarantees
$G_{\mathrm{LP}}(x;\omega)\le\sum_k\alpha_k^{\mathrm{rep}}+\lambda^\top x,\quad\forall x$.
Therefore, every generated cut remains valid even when column generation stops before full price-out; the repair becomes inactive when the unrestricted LP optimum is reached.
}

\REVII{\textbf{Convergence of the LP-relaxed Benders procedure.}
\label{rem:convergence}
The stabilization mechanisms above preserve the validity of the Benders framework rather than serving as heuristic modifications. When the subproblem solver returns dual-feasible recourse information, every cut added to the master remains a valid supporting hyperplane of the recourse value function, regardless of how many iterations the algorithm runs.
A formal statement and proof are provided in Appendix~\ref{app:proofs}.
Therefore, the reported termination metric
$\mathrm{gap}_{\mathrm{Bend}}$
in Section~\ref{sec:numerical study} represents a meaningful optimality measure at whatever iteration the algorithm is stopped, rather than a heuristic stopping criterion.}

\REFTWO{%
\subsection{Profitability-Based Request Pre-screening}
\label{sec:prescreening}

Besides improving convergence, computational effort can also be reduced by
removing requests that are provably unprofitable under a conservative
dedicated-capacity estimate.

For each request $k\in K^\omega$, define the minimum \textit{dedicated}
reservation cost of a feasible path as
\begin{align}
\underline{c}_k:=
\min_{p\in\mathcal{P}_k}
\sum_{(\ell,t):\,a_{kp}^{\ell t}=1} c_{\ell t},
\label{eq:min_cost_path}
\end{align}
i.e., the cost of reserving one dedicated unit of capacity along the cheapest
time-feasible path of request $k$. If $r_k<\underline{c}_k$, request $k$
cannot generate positive net value when served using capacity reserved
exclusively for itself. We therefore optionally fix
\begin{align}
y_{kp}^\omega=0,\quad \forall p\in\mathcal{P}_k,\qquad
z_k^\omega=0,
\qquad \text{whenever } r_k<\underline{c}_k,
\label{eq:prescreen}
\end{align}
before solving the model.

This is an optional pre-screening procedure rather than an
objective-preserving reduction. Because corridor--time capacity is shared, a
screened request may still have positive marginal value when it can exploit
capacity reserved for other requests. The rule therefore trades a potentially
small loss in solution quality for reductions in the number of requests,
column pools, and pricing dimensions. Numerical experiments show that at the
operating reservation price, it removes no requests and introduces no
observable objective loss, while only eliminating a non-trivial fraction of
requests under high-price regimes where the optimal plan already reserves
little capacity.}

\begin{example}[Example: Profitability-based pre-screening]
\REFTWO{Reusing the network of Figure~\ref{fig:toy_network_time}, suppose request $k$ has origin node $1$, destination node $6$, release time $t$, and revenue $r_k=6$. Suppose $\mathcal{P}_k$ contains two feasible paths: $p_1=1\to2\to5\to6$, occupying $(\ell_{12},t),(\ell_{25},t+1),(\ell_{56},t+2)$ at reservation costs $c_{\ell_{12},t}=2$, $c_{\ell_{25},t+1}=3$, and $c_{\ell_{56},t+2}=2.5$, for a total of $7.5$; and a second, costlier path $p_2$ with total reservation cost $9$. Path $p_1$ is cheaper, so $\underline{c}_k=7.5$. Since $r_k=6<\underline{c}_k=7.5$, even reserving \textit{dedicated}, exclusive capacity along the single cheapest path would cost more than $k$ could ever generate in revenue --- sharing capacity with other requests can only add contention, never lower this floor. Rule~\eqref{eq:prescreen} therefore fixes $y_{kp}^\omega=0$ for all $p\in\mathcal{P}_k$ and $z_k^\omega=0$ before the model is solved, removing $k$ from every downstream pricing and restricted-master computation. Had the revenue instead been $r_k=9\ge\underline{c}_k$, request $k$ would survive pre-screening --- though it might still be rejected later, endogenously, if capacity actually binds.}
$\hfill \blacksquare$
\end{example}

\medskip
\section{Computational Study}\label{sec:numerical study}

\REVII{This section reports the computational study. Section~\ref{sec:benchmarks} describes the experimental setup, the compared benchmark methods, and the evaluation metrics. Section~\ref{sec:exp1} reports the scaling performance of \DBCG and the scalable \DBTP heuristic against baselines. Section~\ref{sec:sensitivity_analysis} examines how network connectivity, routing flexibility, route overlap, and scenario sample size shape solution quality and computational difficulty.}

\subsection{Experimental Setup, Benchmarks, and Reporting Protocol}
\label{sec:benchmarks}

\REFONE{\textbf{Experimental setup.}
All algorithms are implemented in Python with Gurobi as the LP/MIP solver and run on a single 16-core workstation, with a single thread per LP solve so that wall-clock times are comparable across runs. Instances with $|\mathcal{V}| \le 16$ were solved with a one-hour time limit, whereas larger instances were allocated a two-hour computational budget. Full software versions, hardware specification, and solver tolerances are given in Appendix~\ref{app:setup}.

}

\REVII{\textbf{Instance generation.}
Unless otherwise stated, experiments use synthetic corridor networks on a jittered lattice, with congestion controlled by three parameters, route overlap $\eta_{\mathrm{ov}}$, capacity tightness $\tau_{\mathrm{cap}}$, and release-window ratio $W$; scaling experiments increase network size, horizon, demand volume, scenario count, and path-pool size jointly, while structural experiments vary only $\eta_{\mathrm{ov}}$. The full generator is described in Appendix~\ref{app:setup}.}

\REVII{
\textbf{Compared methods.}
We compare Benders decomposition, cut aggregation, scalable path
approximation, and non-decomposition baselines to isolate the effects of
scenario disaggregation, recourse approximation, and stochastic planning.
\begin{itemize}
\item \textit{\DBCG (Disaggregated Benders with Column Generation).}
The proposed disaggregated Benders algorithm solves each scenario
recourse problem by column generation and applies the stabilization
strategies described in Section~\ref{sec:acceleration}.
\item \textit{\ABCG (Aggregated Benders with Column Generation).}
This variant replaces scenario-specific cuts with a single expected
optimality cut: $\sum_{\omega\in\Omega}\pi_\omega\theta_\omega \leq \sum_{\omega\in\Omega}\pi_\omega ((\lambda^{\omega *})^\top x+\beta^{\omega *})$. It preserves the same recourse formulation and
subproblem solver, isolating the benefit of cut disaggregation.
\item \textit{\DBTP (Disaggregated Benders with truncated-path recourse).}
For large instances, we replace exact column generation with a
truncated-$P$ path pool while retaining the same outer Benders framework.
This approximation avoids pricing iterations and enables larger-scale
experiments.
\item \textit{Greedy allocation.}
Greedy capacity allocation is a non-decomposition heuristic that first computes unconstrained shortest-path usage, ranks corridor--time pairs by expected utilization, and allocates capacity to the most frequently used pairs; it ignores stochastic competition among requests and the marginal value of capacity, serving as a myopic operational baseline.
\item \textit{\DET (Deterministic mean-demand benchmark).}
This deterministic benchmark solves the problem under mean demand
and evaluates the resulting capacity plan over the stochastic scenarios.
The resulting performance difference quantifies the value of stochastic
planning.
\end{itemize}
}

\REVII{\textbf{Evaluation metrics.}
The reported deviations arise from different sources and should not be
interpreted interchangeably. Table~\ref{tab:metrics} summarizes the
diagnostics used throughout the computational study.}

\begin{table}[!h]
\centering
\caption{Deviation diagnostics used in the computational study.}
\label{tab:metrics}
\begin{tabular}{lll}
\toprule
\textbf{Source} & \textbf{Metric} & \textbf{Purpose} \\
\midrule

Algorithmic convergence
&
$\displaystyle
\mathrm{gap}_{\mathrm{Bend}} (\%)=
\frac{\mathrm{UB}-\mathrm{LB}}{\mathrm{UB}} \times 100
$
&
Benders termination quality
\\[1.5ex]


Recourse approximation
&
$\displaystyle
\mathrm{gap}_{\mathrm{cert}}(\%)=
\frac{\mathrm{UB}-z_{\DBTP}}{\mathrm{UB}}\times 100
$
&
Truncation uncertainty
\\[1.5ex]

Second-stage relaxation
&
Integrality diagnostics
&
LP--MIP discrepancy
\\

\bottomrule
\end{tabular}
\end{table}

\REVII{
These diagnostics capture different sources of deviation. 
The Benders termination gap reflects incomplete convergence of the
decomposition algorithm under a finite computational budget. 
For the scalable \DBTP variant, the certificate gap is not an algorithmic
termination measure but quantifies the approximation introduced by restricting
the recourse problem to a truncated path pool. 
Second-stage integrality is reported separately because it concerns
the relaxation quality of the recourse model rather than the performance of the decomposition algorithm.
}

\subsection{Algorithm Performance}
\label{sec:exp1}
This experiment evaluates whether scenario disaggregation improves convergence when the recourse dimension grows and reports scaling performance (Section~\ref{sec:scaling_exp}) and the scalable \DBTP variant (Section~\ref{sec:exp_heuristics}).

\subsubsection{Scaling and the value of the stochastic solution.}
\label{sec:scaling_exp}

This experiment evaluates scalability while holding the structural regime constant: $\eta_{\mathrm{ov}}=0.70$, $\tau_{\mathrm{cap}}=0.55$, $W=0.25$.
We construct a balanced scaling ladder in which spatial size $|\mathcal{V}|$ increases jointly with the time horizon $|\mathcal{T}|$, request volume $|K^\omega|$, scenario count $|\Omega|$, and path-pool size $P$, scaled proportionally to preserve comparable load intensity and congestion geometry. Thus, the experiment increases both master and subproblem dimensions while controlling for structural effects.

\begin{table}[!htbp]
\centering
\caption{\REVII{Scaling performance of the proposed \DBCG and benchmark solutions.}}
\label{tab:exp1_combined}
\resizebox{1.0\textwidth}{!}{
\begin{tabular}{cccccc}
\toprule
$|\mathcal{V}|$ & \DBCG & Greedy & \DET & $\mathrm{gap}_{\mathrm{Bend}}(\%)$ & Benders iters / Time (s) \\
\midrule
12 & $108.25\pm8.97$ & $61.63\pm5.50$ & $90.34\pm11.04$ & 0.18 & 357 / 3604 \\
14 & $108.58\pm7.92$ & $65.75\pm4.11$ & $89.99\pm10.40$ & 0.47 & 281 / 3611 \\
16 & $120.60\pm10.13$ & $73.50\pm5.34$ & $97.60\pm11.17$ & 1.49 & 72 / 3621 \\
20 & $132.20\pm7.73$ & $86.77\pm4.47$ & $103.85\pm13.50$ & 3.51 & 62 / 7236 \\
25 & $145.37\pm9.71$ & $105.05\pm6.69$ & $114.38\pm10.03$ & 4.83 & 31 / 7441 \\
30 & $145.31\pm5.89$ & $116.00\pm3.12$ & $119.65\pm10.82$ & 21.76 & 8 / 7744 \\
\bottomrule
\end{tabular}
}
\end{table}

\REVII{
Table~\ref{tab:exp1_combined} reports the scaling performance of the proposed
\DBCG approach.
The proposed algorithm reaches small LP-Benders termination gaps through
$|\mathcal V|=25$ ($0.18$--$4.83\%$), with the gap increasing to
$21.76\%$ only at $|\mathcal V|=30$, where the enlarged recourse admits only
eight Benders iterations within the two-hour time budget. These results
indicate that the proposed decomposition remains effective as both the
planning and operational dimensions increase simultaneously.

From a planning perspective, \DBCG consistently produces substantially higher
expected profit than the Greedy allocation baseline, with improvements
ranging from $75.6\%$ at $|\mathcal V|=12$ to $25.3\%$ at
$|\mathcal V|=30$. Although this advantage decreases as the network grows,
the trend is expected because larger instances generate smoother aggregate
utilization patterns, making the utilization-ranking heuristic adopted by
\textit{Greedy} increasingly informative. Nevertheless, Greedy only
captures local utilization patterns and ignores the joint competition among
requests for shared corridor--time capacity. In contrast, \DBCG explicitly
optimizes these stochastic interactions and continues to produce superior
capacity reservation decisions throughout the tested range.
}

\REVII{
The \DET benchmark isolates the value of explicitly modeling demand uncertainty
in strategic capacity reservation. Across all tested scales, the stochastic
formulation solved by \DBCG consistently outperforms \DET, with the value of
the stochastic solution ranging from $16.5\%$ to $21.5\%$. The improvement
arises from the temporal and spatial rigidity of reserved corridor--time
capacity. By optimizing only for mean demand, \DET smooths demand fluctuations
and tends to underestimate the capacity required during peak realizations.
Once demand is realized, such shortages cannot be corrected through temporal
or spatial substitution, leading to irreversible profit losses. These results
demonstrate that uncertainty is not merely an operational disturbance but a
fundamental driver of effective corridor--time capacity planning.
}

\subsubsection{Scalable \DBTP via truncated-path recourse approximation}
\label{sec:exp_heuristics}

Although \DBCG remains effective to medium-sized instances, solving every scenario subproblem to optimality by column generation becomes increasingly expensive for large networks because of the growing path space. To extend the scalability of the same decomposition framework, we introduce \DBTP by replacing the exact column-generation recourse solver with a truncated-path approximation. Specifically, \DBTP solves the same disaggregated Benders master problem as \DBCG, but restricts each scenario subproblem to a fixed pool of $P=8$ shortest feasible paths. This approximation eliminates pricing iterations and substantially reduces computational effort while preserving the joint optimization structure between first-stage capacity reservation and second-stage routing decisions.

Because the truncated subproblem no longer represents the full recourse problem, the resulting Benders bounds are not valid for the original model. We therefore apply the multi-fidelity certification procedure, which performs additional pricing passes to restore dual feasibility with respect to the full path set and obtain a valid upper bound for the original problem; the validity of this bound is established in Appendix~\ref{app:proofs}. The resulting certificate gap, $(\mathrm{UB}-z_{\mathrm{DBTP}})/\mathrm{UB}$, quantifies the remaining optimality uncertainty introduced by the truncated recourse approximation.

Unlike $\mathrm{gap}_{\mathrm{Bend}}$, which measures the convergence status of the exact \DBCG decomposition under a finite computational budget, the certificate gap measures the loss caused by approximating the recourse representation. Together, these two metrics distinguish algorithmic termination effects from approximation effects in scalable solution strategies.

\begin{table}[!h]
\centering
\caption{\REVII{Performance of \DBTP under exact recourse evaluation.}}
\label{tab:exp_heuristics_exact}
\begin{tabular}{cccccc}
\toprule
$|\mathcal V|$
& \DBTP
& Greedy
& $\mathrm{Imp}_{\mathrm{Greedy}}(\%)$
& $\mathrm{gap}_{\mathrm{cert}}(\%)$
& Time (s)
\\
\midrule
8  & 74.1  & 52.5  & +41.1 & 22.3 & 121\\
12 & 97.5  & 64.3  & +51.8 & 22.3 & 121\\
16 & 107.1 & 74.7  & +43.4 & 23.8 & 121\\
20 & 121.7 & 82.1  & +48.2 & 24.7 & 121\\
\bottomrule
\end{tabular}
\end{table}

\begin{table}[!htbp]
\centering
\caption{\REVII{Scalability of \DBTP under truncated-path recourse evaluation.}}
\label{tab:exp_heuristics_trunc}
\begin{tabular}{cccccc}
\toprule
$|\mathcal V|$
& \DBTP
& Greedy
& $\mathrm{Imp}_{\mathrm{Greedy}}(\%)$
& $\mathrm{gap}_{\mathrm{cert}}(\%)$
& Time (s)
\\
\midrule
30 & 125.3 & 74.2  & +68.9 & 39.5 & 603\\
40 & 124.4 & 83.4  & +49.2 & 45.6 & 604\\
50 & 179.2 & 124.2 & +44.2 & 44.9 & 605\\
60 & 214.5 & 152.9 & +40.3 & 47.2 & 605\\
70 & 224.6 & 184.5 & +21.8 & 57.6 & 607\\
\bottomrule
\end{tabular}
\end{table}

Tables~\ref{tab:exp_heuristics_exact} and~\ref{tab:exp_heuristics_trunc} evaluate the scalability of the proposed Benders framework when the exact column-generation recourse solver is replaced by a truncated-path approximation. Table~\ref{tab:exp_heuristics_exact} considers instances for which the original LP recourse remains computationally evaluable, allowing the solution quality of \DBTP to be assessed against the original recourse model. For larger networks with larger scenario sets, request volumes, and path pools, exact recourse evaluation becomes computationally prohibitive. Table~\ref{tab:exp_heuristics_trunc} reports results under a common truncated-$P$ operational recourse evaluation. Therefore, objective values across the two tables should not be directly compared; each table should be interpreted within its corresponding evaluation setting. 

The purpose of this experiment is to examine whether replacing the exact column-generation recourse solver with a truncated-path approximation can extend the scalability of the proposed Benders framework without substantially sacrificing planning quality. We evaluate \DBTP from three perspectives: solution quality relative to the Greedy benchmark, scalability to larger networks, and the reliability of the resulting approximation. As a result, we identify the following findings: 
First, \DBTP preserves the main planning advantages of \DBCG while reducing the computational burden of repeated pricing iterations. Under exact recourse evaluation, \DBTP improves over the greedy allocation baseline by $41$--$52\%$, showing that a small truncated path pool can retain substantial value from joint capacity reservation and routing optimization. For moderate network sizes, its solutions remain close to those obtained by the exact \DBCG. 
Second, \DBTP enables the framework to handle larger networks. Under the truncated-$P$ evaluation, the improvement over greedy remains above $40\%$ for $|\mathcal V|\le60$. The decline at $|\mathcal V|=70$ reflects the limitation of using a fixed path-pool size, as larger networks introduce additional valuable routing alternatives. Because exact recourse evaluation becomes computationally prohibitive beyond $|\mathcal V|=20$ (Table~\ref{tab:exp_heuristics_exact}), we cannot directly verify the magnitude of capacity over-reservation at these larger scales; Section~\ref{sec:exp_pool_truncation} shows this effect reaches $26$--$67\%$ at $|\mathcal V|\le30$ under the same restricted-pool mechanism, so the $|\mathcal V|=40$--$70$ plans reported here likely carry a comparable, unverified degree of over-reservation.
Third, the multi-fidelity procedure provides an explicit quality assessment for \DBTP. The certificate gap increases from approximately $22$--$25\%$ for small instances to larger values at larger scales, indicating that a fixed path pool becomes less representative as routing flexibility grows. These certificates can be tightened through additional pricing passes, providing a controllable trade-off between computational efficiency and solution guarantees.

\subsection{Sensitivity Analysis}
\label{sec:sensitivity_analysis}
These experiments are designed to isolate how uncertainty, topology, and reservation economics affect capacity reservation decisions.

\REVII{\subsubsection{Impact of Network Connectivity on Scalability}
\label{sec:path_degeneracy}

The scalability of \DBCG depends not only on the size of the corridor--time network, but also on its connectivity. In the proposed formulation, each request is represented by a set of feasible time-respecting paths, and the number of such paths grows rapidly as network connectivity increases. Consequently, highly connected networks may introduce many near-equivalent routing alternatives, enlarging the restricted master problem and making the column-generation process more challenging. This experiment investigates how network connectivity affects the computational performance of the CG solver and identifies the structural source of this sensitivity.

We vary the average node degree while keeping the network size ($|\mathcal V|=12$) and other experimental parameters fixed. Table~\ref{tab:connectivity} reports the resulting CG iterations, generated columns, pricing time, and total solution time. It indicates that increasing connectivity substantially increases the computational burden of column generation. As the average degree rises from $2$ to $6$, the number of CG iterations increases from $11$ to $66$, while the generated column pool grows from $138$ to $1653$. The resulting total solution time increases from $12.7$s to $105.2$s. 
}

\begin{table}[!h]
\centering
\caption{Effect of network connectivity on column-generation performance.}
\label{tab:connectivity}
\begin{tabular}{ccccc}
\toprule
Average degree
& CG iterations
& Generated columns
& Pricing time (s)
& Total time (s)\\
\midrule
2.0 & 11 & 138 & 4.3 & 12.7\\
3.0 & 18 & 296 & 7.5 & 21.6\\
4.0 & 31 & 621 & 12.1 & 39.8\\
5.0 & 47 & 1018 & 18.6 & 67.5\\
6.0 & 66 & 1653 & 28.7 & 105.2\\
\bottomrule
\end{tabular}
\end{table}

\REVII{
The observed trend is explained by the increasing number of near-equivalent time-feasible paths in denser networks. These alternative paths often have similar reduced costs, leading to stronger path degeneracy in the restricted master problem and slower dual stabilization. Importantly, this effect is geometric rather than combinatorial: it arises from multiple similar routing alternatives rather than symmetry among vehicle or route assignments. The pricing problem itself remains computationally tractable because it is solved as a longest reduced-cost path problem on the acyclic time-expanded network, with complexity $O((|\mathcal V|+|\mathcal L|)|\mathcal T|)$ per request.}

\REVII{\subsubsection{Impact of Routing Flexibility on Strategic Capacity Planning}
\label{sec:exp_pool_truncation}

An important structural question is how operational routing flexibility affects strategic capacity planning. In the proposed stochastic corridor--time capacity planning problem, capacity is reserved before demand uncertainty is resolved, whereas routing decisions are adapted after demand realization. Therefore, the value of reserved corridor--time capacity depends not only on the amount of capacity available, but also on the ability of the recourse model to exploit alternative time-respecting paths and share capacity across requests. To quantify the strategic value of routing flexibility, 
we compare two self-consistent planning regimes. In the \textit{complete} regime, capacity is planned and evaluated with each scenario recourse solved exactly by column generation over the full path set. In the \textit{restricted} regime, capacity is planned and evaluated with every request confined to a fixed pool of $P=8$ shortest feasible paths. The comparison is therefore between two internally coherent operating models, one with full routing flexibility in both planning and execution and one with restricted flexibility in both, rather than a decomposition of the effect into a first-stage and a second-stage component. The reservation expenditure is a first-stage quantity and is
directly comparable across regimes; the objective difference additionally reflects the narrower recourse used to score the restricted plan and should be read as an upper estimate of the loss attributable to the first-stage decision alone. This experiment tests the independently generated instances ensembling at each $|\mathcal V|$.}

\begin{table}[!h]
\centering
\caption{\REVII{Impact of recourse flexibility on strategic capacity planning:
expected objective values and capacity reservation costs under complete and
restricted path representations.}}
\label{tab:exp_pool_truncation}
\begin{tabular}{ccccccc}
\toprule
$|\mathcal V|$
& \multicolumn{2}{c}{Expected objective}
& Profit loss
& \multicolumn{2}{c}{Capacity cost}
& Cost increase
\\
&
Complete
& Restricted
& (\%)
&
Complete
& Restricted
& (\%)
\\
\midrule
12 &108.1&90.0&16.7&16.9&21.2&25.7\\
14 &108.2&89.7&17.1&17.2&23.0&34.0\\
16 &117.1&97.9&16.3&19.3&24.6&27.5\\
20 &132.0&105.9&19.8&23.6&32.9&39.1\\
25 &144.0&110.6&23.2&24.8&41.5&67.2\\
30 &153.5&121.9&20.6&31.6&47.2&49.6\\
\bottomrule
\end{tabular}
\end{table}

\REVII{The results demonstrate the strategic value of recourse flexibility in stochastic corridor--time capacity planning. Restricting the path representation reduces the expected objective value by approximately $16$--$23\%$ across all tested instances, while increasing capacity reservation cost by $26$--$67\%$. Thus, limiting routing flexibility not only reduces operational effectiveness after demand realization, but also distorts the first-stage capacity reservation decision. The mechanism is that a richer recourse path set allows realized requests to exploit alternative corridor--time combinations and share reserved capacity more effectively. When feasible routing alternatives are removed, requests become more dependent on specific capacity units, causing the planner to reserve additional capacity to hedge against congestion. However, this additional investment cannot fully compensate for the loss of routing adaptability, leading to lower expected profit. 

In a word,  these results show that routing flexibility is an integral component of strategic capacity planning rather than merely an operational detail. The capacity value of a corridor--time reservation depends critically on the future ability to reallocate demand through alternative routes.}

\REVII{
\subsubsection{Effect of Route Overlap on Decomposition Convergence}
\label{sec:exp_structural}

Problem difficulty in stochastic corridor--time capacity planning is driven not only by network size, but also by the congestion structure induced by competing requests. Route overlap determines how strongly different requests rely on shared corridor--time resources. On one hand, when overlap is low, requests are spatially separated and scenario subproblems tend to generate similar capacity values. On the other hand, when overlap is high, many requests compete for the same limited resources, creating stronger coupling and more diverse scenario-specific marginal values of capacity.

This structural variation is particularly relevant for Benders decomposition because the quality of an optimality cut depends on how accurately it represents the scenario-specific value of reserved capacity. To isolate the benefit of scenario disaggregation, we compare \DBCG and \ABCG, which use the same column-generation recourse solver and stabilization strategies but differ only in whether scenario cuts are retained separately or aggregated into a single expected cut. 
Specifically, we vary the route-overlap parameter $\eta_{\mathrm{ov}}$ while fixing $|\mathcal V|=20$, $|\mathcal T|=40$, $K=30$, $|\Omega|=6$, $P=8$, $\tau_{\mathrm{cap}}=0.55$, and $W=0.25$ with a one-hour time limit.
}

\begin{table}[!h]
\centering
\caption{\REVII{Effect of route overlap on the convergence performance of \DBCG and \ABCG.}}
\label{tab:exp2_overlap}
\begin{tabular}{ccccc}
\toprule
& \multicolumn{2}{c}{\DBCG}
& \multicolumn{2}{c}{\ABCG}
\\
\cmidrule(lr){2-3}\cmidrule(lr){4-5}
$\eta_{\mathrm{ov}}$
& Objective
& Gap (\%)
& Objective
& Gap (\%)
\\
\midrule
0.40
& 89.25
& 10.4
& 89.67
& 10.7
\\
0.60
& 97.21
& 10.3
& 94.09
& 14.0
\\
0.95
& 103.43
& 15.1
& 101.53
& 18.2
\\
\bottomrule
\end{tabular}
\end{table}

\REVII{
Table~\ref{tab:exp2_overlap} shows that the advantage of scenario-disaggregated cuts becomes more pronounced as route overlap increases. The reported gaps are the residual $\mathrm{gap}_{\mathrm{Bend}}$ values defined in Section~\ref{sec:benchmarks}. 
Under low overlap, \DBCG and \ABCG achieve similar convergence behaviour because requests compete weakly for shared corridor--time resources and scenario-specific capacity values are relatively aligned. As overlap increases, the residual gap of \ABCG grows more rapidly than that of \DBCG, indicating that aggregated cuts become less effective when different scenarios induce heterogeneous congestion patterns. 

The mechanism is that high route overlap amplifies the difference between scenario-specific marginal capacity values. In \ABCG, these distinct dual signals are averaged into a single expected cut, which provides a weaker approximation of the recourse function at the current master solution. In contrast, \DBCG preserves scenario-level information and can adapt capacity allocation to heterogeneous congestion conditions, leading to more effective master updates.
}

\subsubsection{Scenario Sample Size and SAA Solution Stability}
\label{sec:exp_saa}

\REVII{
Because the stochastic program relies on a finite scenario sample to approximate the underlying demand distribution, the resulting capacity plan may depend on the scenarios used during optimization. We therefore evaluate the statistical stability of SAA-based planning using independent training and testing scenario sets. The training size varies as $|\Omega_{\mathrm{in}}|\in\{10,25,50,100\}$, while $|\Omega_{\mathrm{out}}|=200$ is fixed for out-of-sample evaluation. For each setting, 20 independent replications are performed. All SAA problems are solved using the same Benders termination criterion to isolate the effect of scenario sampling from algorithmic convergence. Stability is evaluated through out-of-sample profit, generalization gap, and capacity-plan variability. }

\begin{table}[!htbp]
\centering
\caption{SAA out-of-sample stability for a representative small instance ($|\mathcal V|=12$)}
\begin{tabular}{ccccc}
\toprule
$|\Omega_{\mathrm{in}}|$
&
$\mathbb E[V_{\mathrm{out}}]$
&
Std$(V_{\mathrm{out}})$
&
$\Delta_{\mathrm{gen}}$
&
CV$(\|x\|_1)$
\\
\midrule
10
&94.6
&8.5
&5.8
&0.31
\\

25
&98.7
&5.2
&3.1
&0.19
\\

50
&101.2
&3.1
&1.6
&0.12
\\

100
&102.0
&2.0
&0.8
&0.08
\\
\bottomrule
\end{tabular}
\end{table}

\REVII{
At $|\mathcal V|=12$, all stability measures improve as the training sample size increases. The average out-of-sample profit increases from $94.6$ to $102.0$, while its standard deviation decreases from $8.5$ to $2.0$, indicating that larger scenario samples produce both better and more predictable capacity plans. The generalization gap $\Delta_{\mathrm{gen}}$, defined as the relative difference between the in-sample SAA objective and the out-of-sample evaluation, decreases from $5.8\%$ to $0.8\%$, showing that small scenario samples lead to more sample-dependent planning decisions. The coefficient of variation of total reserved capacity, $\mathrm{CV}(\|x\|_1)$, follows the same trend and decreases from $0.31$ to $0.08$, demonstrating that scenario enrichment stabilizes not only the evaluated profit but also the resulting capacity allocation. However,  such improvements become smaller beyond $|\Omega_{\mathrm{in}}|=50$, suggesting that moderate scenario samples already provide a favorable balance between solution stability and computational effort for the tested instances.
}


\medskip
\section{Case Study: Capacity Planning for Urban Low-Altitude Logistics in Shenzhen}\label{sec:case study}
The preceding experiments used synthetic corridor networks to isolate specific structural drivers of algorithmic and planning performance under controlled conditions. This section grounds those findings in a real urban airspace topology.

\subsection{Shenzhen UAV Network}
\label{sec:case_network}

Shenzhen has become one of the largest commercial UAV logistics testbeds in China, with hundreds of designated routes and distributed take-off and landing facilities. At this operational scale, urban airspace functions as a regulated infrastructure system rather than an open-access environment: corridors are explicitly designated, stations are geographically distributed, and operations are centrally coordinated under strict safety separation rules.
Corridor--time resources are therefore finite, regulated, and potentially congested.

This operational context aligns directly with the planning problem studied in this paper. Capacity must be reserved in advance on designated corridors, while stochastic delivery requests compete for limited corridor--time slots during execution. Rather than reproducing the full-scale Shenzhen UAV traffic system, the case study serves as a structural validation of the proposed capacity planning framework under a realistic urban airspace topology.

\textbf{Network construction.} Figure~\ref{fig:shenzhen_longhua} shows 13 publicly reported UAV take-off and landing stations in Longhua District, Shenzhen.  These sites form the node set of the corridor network in our case study.
We construct directed arcs between any pair of stations whose geographic distance does not exceed the operational flight radius of the UAV fleet, further filtered by geographic feasibility and operational regulations; the resulting candidate corridor set is a sparse directed network, not a complete graph.
The resulting airspace is modeled as a directed graph $G=(\mathcal{V},E)$, where $\mathcal{V}$ denotes stations and $E$ denotes feasible flight corridors.
Time is discretized into operational intervals consistent with regulatory separation standards.
A delivery request from origin to destination occupies a sequence of corridor--time resources along a feasible path in the time-expanded network.
Each corridor--time pair admits a limited number of flight slots, representing regulated capacity units.
These units constitute the first-stage decision variables in the capacity planning problem.

\begin{figure}[htbp]
    \centering
    \includegraphics[width=0.6\linewidth]{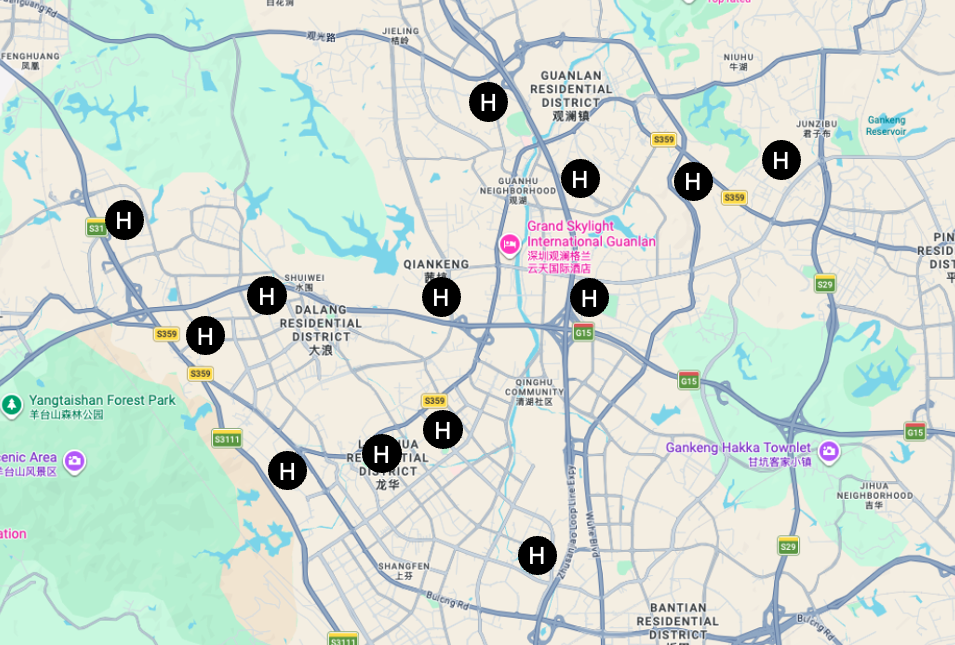}
    \caption{Spatial distribution of 13 operational UAV stations in Longhua District, Shenzhen.}
    \label{fig:shenzhen_longhua}
\end{figure}

\subsection{Capacity--Demand Trade-offs under Stochastic Planning}
\label{sec:case_results}

\textbf{Demand characteristics}.
Demand is synthetically generated to capture key characteristics of urban UAV logistics, including spatially distributed origins and destinations and time-varying request arrivals.
To represent demand uncertainty, six representative demand scenarios are generated to capture plausible spatial and temporal variability in urban delivery requests.
Each scenario specifies a realization of requests over the planning horizon and is used in the second-stage recourse problem to evaluate first-stage capacity decisions.

To analyze system behavior under varying congestion intensity, we scale the baseline demand to construct low, medium, and high load regimes, corresponding to approximately $10$, $20$, and $30$ requests per scenario, respectively, over a planning horizon of $|\mathcal T|=20$ periods.
Higher load levels increase the number of concurrent requests and intensify competition for corridor--time slots.
We further introduce a capacity cost multiplier $\kappa$ that scales the unit reservation cost of corridor--time capacity to $0.05\kappa$ per unit.
The parameter $\kappa$ governs the trade-off between upfront investment and recourse revenue: higher values increase reservation cost and reduce optimal capacity levels.
The planner must therefore determine the amount of corridor--time capacity to reserve in advance, balancing expected routing revenue against reservation expenditure under stochastic demand.

\begin{table}[!h]
\centering
\caption{\REVII{Capacity allocation and acceptance performance of the proposed \DBCG under varying demand loads and cost multipliers $\kappa$.}}
\label{tab:benders_core}
\begin{tabular}{ccccc}
\toprule
Load & $\kappa$ & Objective & Accepted Requests & Capacity Cost \\
\midrule
High   & 0.5 & $49.81\pm1.12$ & 16.53 & 8.96 \\
High   & 1.0 & $40.13\pm1.35$ & 15.88 & 16.42 \\
High   & 2.0 & $25.70\pm0.98$ & 12.42 & 20.34 \\
High   & 4.0 & $8.66\pm0.42$  & 3.11  & 6.55 \\
\midrule
Medium & 0.5 & $31.26\pm0.86$ & 10.23 & 6.52 \\
Medium & 1.0 & $24.03\pm1.04$ & 9.17  & 12.28 \\
Medium & 2.0 & $12.88\pm0.51$ & 5.43  & 14.84 \\
Medium & 4.0 & $4.78\pm0.29$  & 1.89  & 4.65 \\
\midrule
Low    & 0.5 & $14.65\pm0.74$ & 5.79  & 3.80 \\
Low    & 1.0 & $10.31\pm0.63$ & 4.54  & 6.96 \\
Low    & 2.0 & $5.13\pm0.21$  & 2.30  & 6.82 \\
Low    & 4.0 & $1.50\pm0.18$  & 0.60  & 2.71 \\
\bottomrule
\end{tabular}
\end{table}

\REVII{\textbf{Computational results}. 
Table~\ref{tab:benders_core} reports the solutions obtained by \DBCG on the Longhua case under different demand loads and reservation cost multipliers $\kappa$. Each entry is reported as mean $\pm$ standard deviation over three independent demand seeds. ``Capacity Cost'' denotes the total reservation expenditure $\boldsymbol c^\top\boldsymbol x^*$, while ``Accepted Requests'' reports the average number of accepted requests across scenarios and seeds. The results reveal two key patterns.

First, increasing $\kappa$ systematically reduces both accepted requests and expected profit across all demand regimes, reflecting the trade-off between reservation expenditure and service revenue. From $\kappa=0.5$ to $\kappa=4$, the objective decreases by approximately $83$--$90\%$, but the planner does not completely withdraw from the market even under the highest reservation cost. For example, under High load, the solution still serves approximately $3.1$ requests per scenario at $\kappa=4$, indicating that the model selectively preserves requests with sufficiently high marginal value. The reported capacity cost is non-monotonic because it combines both the unit reservation price and the amount of reserved capacity: it increases initially as the price multiplier dominates the reduction in reserved units, and decreases only when demand becomes sufficiently unattractive at high reservation costs.

Second, demand intensity determines the scale of profitable capacity investment. Under the lower cost regimes ($\kappa\leq1$), High load produces the largest reservation and acceptance levels (approximately $16$ accepted requests out of $30$ per scenario), followed by Medium and Low loads. This pattern indicates that congestion increases the marginal value of corridor--time capacity: when more requests compete for limited resources, additional reservations generate greater operational benefit.

Finally, we verify that the profitability pre-screening rule~\eqref{eq:prescreen} is inactive in the economically relevant regime of this case. It removes no requests when $\kappa\in\{0.5,1\}$, where the reservation cost remains below the minimum request revenue, and only becomes active in the high-cost regimes where the optimal plan already substantially reduces service volume. This confirms that the screening rule does not affect the practical planning decisions observed in the operating region.
}

\subsection{Spatial Structure of Strategic Capacity Allocation}
\label{sec:case_spatial}

We visualize the resulting corridor--time capacity allocation solutions in Figures~\ref{fig:cap_k2_load}--\ref{fig:cap_med_method}. Across all demand levels, cost parameters, and solution methods, capacity consistently concentrates on a small subset of central corridors, forming a stable backbone network.

Figure~\ref{fig:cap_k2_load} fixes $\kappa=2$ and varies demand load. 
Under high load, capacity is heavily allocated to corridors connecting central stations, with additional peripheral links activated to absorb spillover demand. 
As load decreases, the main reserved corridors remain largely unchanged, while the amount of reserved capacity decreases and peripheral links disappear.
Demand scaling therefore primarily affects allocation magnitude rather than corridor selection.

\begin{figure}[!htbp]
    \centering
    \includegraphics[width=\linewidth]{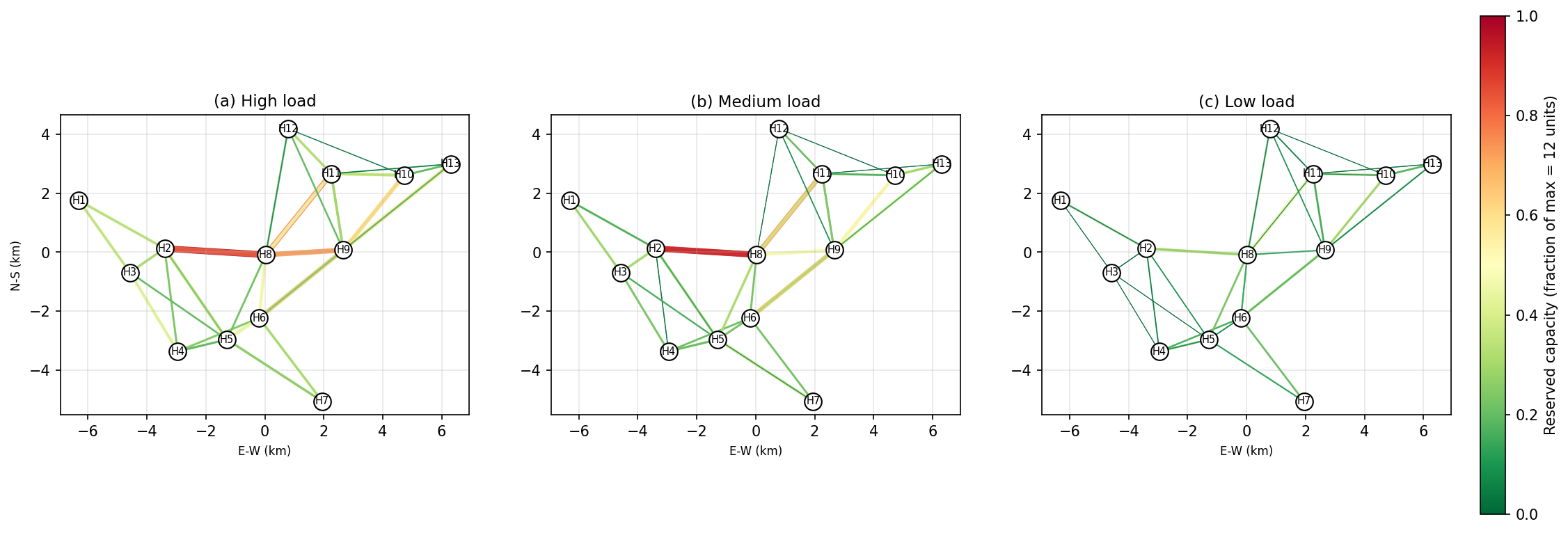}
    \caption{Benders capacity plans under $\kappa=2$ for different demand loads: (a) High, (b) Medium, (c) Low. Edge color indicates log-scaled relative capacity intensity under a common normalization rule.}
    \label{fig:cap_k2_load}
\end{figure}
\begin{figure}[!htbp]
    \centering
    \includegraphics[width=\linewidth]{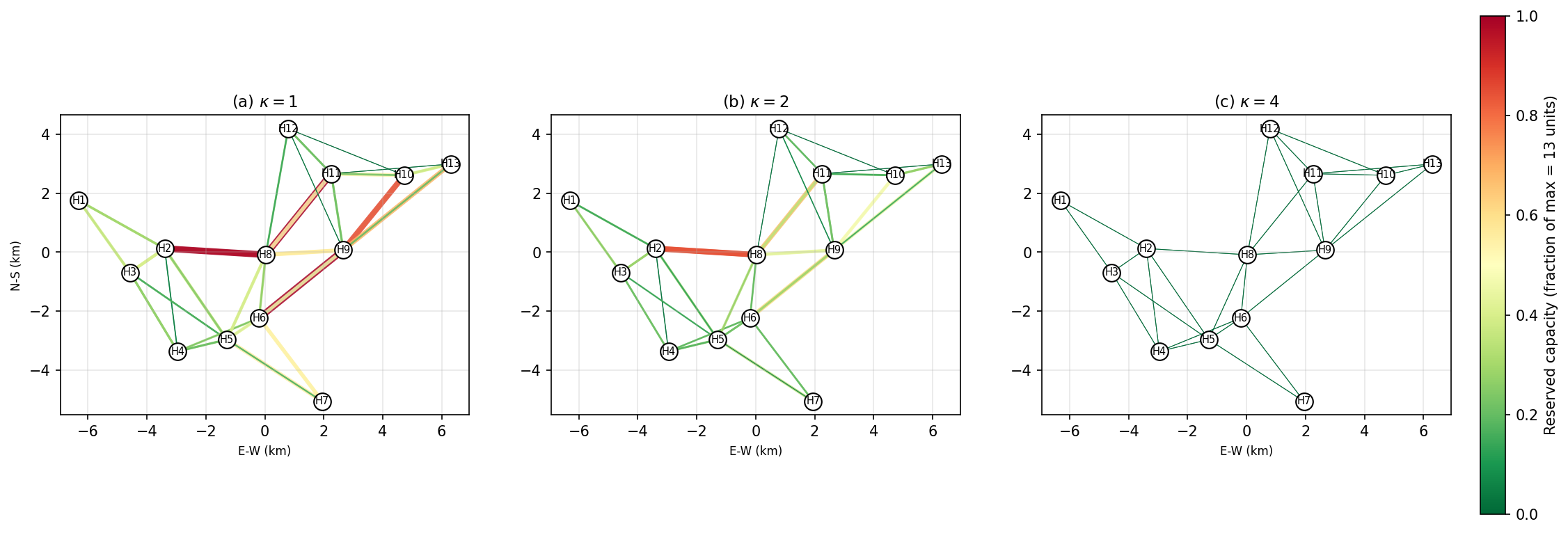}
    \caption{Benders capacity plans under Medium load for $\kappa=1,2,4$. Edge color indicates log-scaled relative capacity intensity under a common normalization rule.}
    \label{fig:cap_med_kappa}
\end{figure}

Figure~\ref{fig:cap_med_kappa} fixes the demand level and varies the reservation cost multiplier $\kappa$. As $\kappa$ increases, capacity allocation gradually contracts from peripheral corridors toward the central backbone, indicating that cost increases first eliminate marginally valuable capacity while preserving structurally important links. At $\kappa=4$, the reserved capacity intensity decreases substantially across the network, consistent with the reduced but nonzero objective value reported in Table~\ref{tab:benders_core}. Nevertheless, the same backbone corridors remain active, suggesting that cost primarily affects how much capacity is reserved rather than which corridors are reserved.

Figure~\ref{fig:cap_med_method} compares capacity plans obtained by different solution methods under identical parameters. The two Benders variants produce similar spatial patterns, concentrating capacity on the same backbone region, which is consistent with their close objective values on this instance. In contrast, greedy allocation concentrates capacity more aggressively on high-utilization corridors rather than distributing it across coordinated routing alternatives. This difference highlights that network structure determines where capacity is valuable, while stochastic optimization determines how capacity should be distributed among competing corridor--time resources.

\begin{figure}[!htbp]
    \centering
    \includegraphics[width=\linewidth]{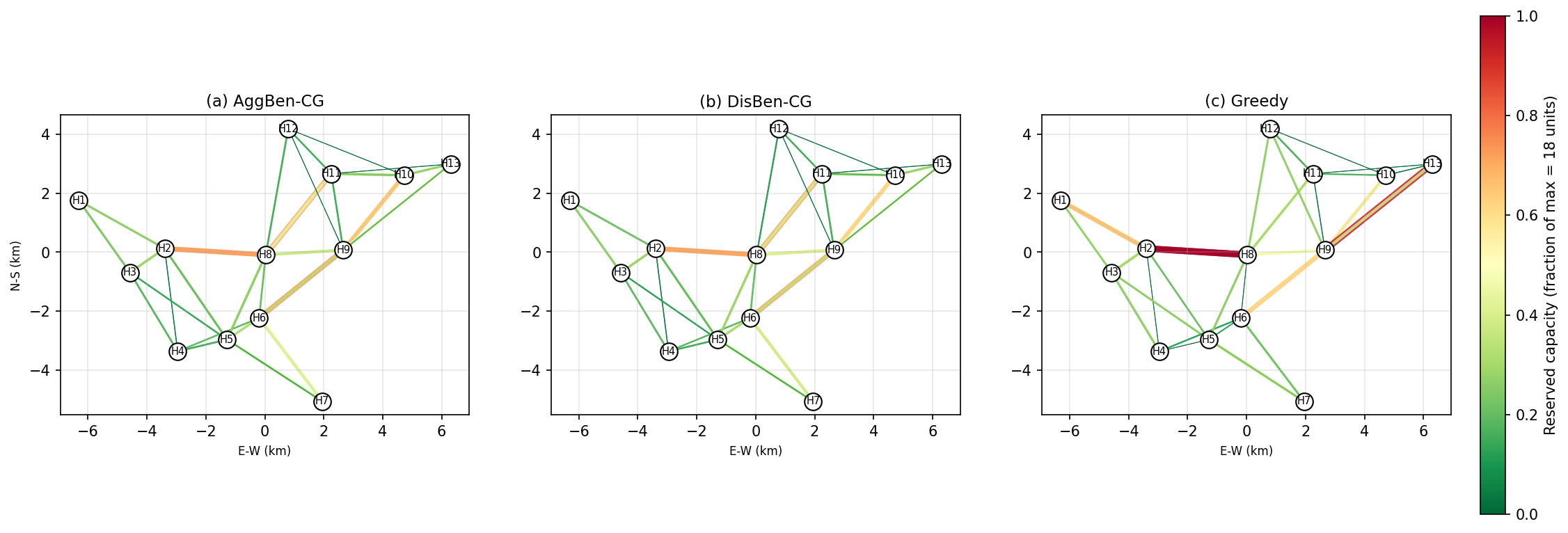}
    \caption{Capacity plans under Medium load and $\kappa=1$ for different methods: (a) \ABCG, (b) \DBCG, and (c) greedy.}
    \label{fig:cap_med_method}
\end{figure}

\REVII{
The spatial patterns reveal that strategic capacity allocation is primarily shaped by the interaction between network geometry and recourse flexibility. Because flight-range constraints produce a sparse corridor network, a limited number of links repeatedly appear in feasible routes across many origin--destination pairs. These structurally central corridors therefore experience persistent recourse pressure and attract capacity investment across different demand levels and cost regimes. Stochastic demand affects the amount of capacity reserved on these corridors, but has a limited impact on which corridors form the backbone of the optimal plan. 

Overall, the Shenzhen case study demonstrates the practical implications of stochastic corridor--time capacity planning in a realistic urban airspace setting. Effective investment requires balancing reservation cost against future demand opportunities: capacity expands when congestion creates sufficient marginal value and contracts when reservation becomes expensive. 
Moreover, coordinated planning preserves routing flexibility and concentrates investment on strategically important corridors, providing a more reliable basis for managing regulated low-altitude logistics networks than purely utilization-driven allocation.}

\medskip
\section{Conclusion}\label{sec:conclusion}

This paper studies stochastic corridor--time capacity planning for regulated low-altitude airspace systems with uncertain delivery demand. Unlike traditional air traffic management approaches that primarily focus on operational adjustments under fixed infrastructure, we formulate capacity planning as a two-stage stochastic optimization problem on a time-expanded corridor network. The planner first reserves corridor--time capacity before demand realization, and subsequently accepts and routes requests subject to the reserved resources. 
The key structural features of low-altitude airspace operations are spatially distributed capacity, temporal rigidity, and competition among requests for shared corridor--time resources.

\REVII{
To solve the resulting stochastic program, we develop \DBCG, a stabilized Benders decomposition in which each scenario recourse problem is solved through column generation. The algorithm preserves scenario-specific recourse information through disaggregated cuts, dual-feasible cut repair, and integer recovery to improve robustness. 
We further develop \DBTP, a scalable truncated-path extension that replaces exact recourse evaluation with a certified approximation for larger networks. The computational analysis shows that \DBCG provides consistently tighter convergence than \ABCG, while the truncated-path framework extends the applicability of stochastic capacity planning beyond the range where exact recourse optimization is computationally practical.
}

\REVII{
The computational study reveals several structural insights. First, preserving scenario-specific congestion information is critical for scalable stochastic planning. \DBCG achieves substantially tighter convergence than \ABCG under the same computational budget, with the advantage becoming more pronounced as scenario interactions and route overlap increase. 
Second, routing flexibility is a strategic component of capacity planning. Restricting the available path set reduces realized profit and induces excessive capacity reservation, showing that capacity investment decisions must anticipate future routing alternatives rather than optimize against a restricted operational representation. 
Third, network connectivity and congestion structure influence both solution quality and computational difficulty. Higher route overlap increases the value of shared corridor--time resources while also creating more challenging recourse structures for decomposition algorithms. 
Fourth, SAA experiments demonstrate diminishing returns from increasing the training scenario set. At the tested scale, most of the improvement in out-of-sample performance and capacity-plan stability is achieved with a moderate number of scenarios.
The Shenzhen case study further validates these findings in a realistic urban airspace topology. The optimal capacity plan consistently concentrates investment on structurally important corridors, while changes in demand and reservation cost primarily affect how much capacity is reserved rather than which corridors are reserved. These results highlight that effective low-altitude airspace management requires coordinated stochastic capacity planning that balances future demand uncertainty, routing flexibility, and infrastructure investment.
}

Several avenues for future research remain. The current framework assumes linear reservation costs and a fixed corridor network. Relaxing these assumptions to allow nonlinear pricing, endogenous network expansion, or evolving airspace infrastructure would provide a richer representation of long-term planning decisions. Future work could also integrate rolling-horizon planning with online demand updates, as well as incorporate operational considerations such as vehicle endurance, battery management, payload limitations, safety separation, and heterogeneous UAV capabilities. These extensions would strengthen the connection between strategic corridor reservation and real-world low-altitude airspace operations.

\section*{Data Availability Statement}
The data supporting the findings of this study were derived from publicly
available sources. The Shenzhen corridor network and demand instances were
constructed by the authors from these sources and are available from the
corresponding author upon reasonable request.

\section*{Conflict of Interest Statement}
The authors declare no conflicts of interest.



\bibliography{ref}

@article{lium2009study,
  title={A study of demand stochasticity in service network design},
  author={Lium, Arnt-Gunnar and Crainic, Teodor Gabriel and Wallace, Stein W},
  journal={Transportation Science},
  volume={43},
  number={2},
  pages={144--157},
  year={2009},
  publisher={INFORMS}
}

@article{zografos2017increasing,
  title={Increasing airport capacity utilisation through optimum slot scheduling: review of current developments and identification of future needs.},
  author={Zografos, Konstantinos and Madas, Michael and Androutsopoulos, Konstantinos},
  journal={Journal of Scheduling},
  volume={20},
  number={1},
  year={2017}
}

@article{bertsimas2016fairness,
  title={Fairness and collaboration in network air traffic flow management: An optimization approach},
  author={Bertsimas, Dimitris and Gupta, Shubham},
  journal={Transportation Science},
  volume={50},
  number={1},
  pages={57--76},
  year={2016},
  publisher={INFORMS}
}

@article{merkert2021will,
  title={Will It Fly? Adoption of the road pricing framework to manage drone use of airspace},
  author={Merkert, Rico and Beck, Matthew J and Bushell, James},
  journal={Transportation Research Part A: Policy and Practice},
  volume={150},
  pages={156--170},
  year={2021},
  publisher={Elsevier}
}

@article{li2022traffic,
  title={Traffic management and resource allocation for UAV-based parcel delivery in low-altitude urban space},
  author={Li, Ang and Hansen, Mark and Zou, Bo},
  journal={Transportation Research Part C: Emerging Technologies},
  volume={143},
  pages={103808},
  year={2022},
  publisher={Elsevier}
}

@article{he2025air,
  title={Air corridor planning for urban drone delivery: Complexity analysis and comparison via multi-commodity network flow and graph search},
  author={He, Xinyu and Li, Lishuai and Mo, Yanfang and Sun, Zhankun and Qin, S Joe},
  journal={Transportation Research Part E: Logistics and Transportation Review},
  volume={193},
  pages={103859},
  year={2025},
  publisher={Elsevier}
}

@article{zhou2020resilient,
  title={Resilient uav traffic congestion control using fluid queuing models},
  author={Zhou, Jiazhen and Jin, Li and Wang, Xiao and Sun, Dengfeng},
  journal={IEEE Transactions on Intelligent Transportation Systems},
  volume={22},
  number={12},
  pages={7561--7572},
  year={2020},
  publisher={IEEE}
}

@article{carraminana2025fair,
  title={A Fair and Congestion-Aware Flight Authorization Framework for Unmanned Traffic Management},
  author={Carrami{\~n}ana, David and Besada, Juan A and Bernardos, Ana M},
  journal={Aerospace},
  volume={12},
  number={10},
  pages={881},
  year={2025},
  publisher={MDPI}
}

@article{kunnen2023cross,
  title={Cross-border capacity planning in air traffic management under uncertainty},
  author={K{\"u}nnen, Jan-Rasmus and Strauss, Arne K and Ivanov, Nikola and Jovanovi{\'c}, Radosav and Fichert, Frank and Starita, Stefano},
  journal={Transportation Science},
  volume={57},
  number={4},
  pages={999--1018},
  year={2023},
  publisher={INFORMS}
}

@article{starita2020air,
  title={Air traffic control capacity planning under demand and capacity provision uncertainty},
  author={Starita, Stefano and Strauss, Arne K and Fei, Xin and Jovanovi{\'c}, Radosav and Ivanov, Nikola and Pavlovi{\'c}, Goran and Fichert, Frank},
  journal={Transportation Science},
  volume={54},
  number={4},
  pages={882--896},
  year={2020},
  publisher={INFORMS}
}

@article{taghavi2016multi,
  title={A multi-stage stochastic programming approach for network capacity expansion with multiple sources of capacity},
  author={Taghavi, Majid and Huang, Kai},
  journal={Naval Research Logistics (NRL)},
  volume={63},
  number={8},
  pages={600--614},
  year={2016},
  publisher={Wiley Online Library}
}

@article{rahmaniani2017benders,
  title={The Benders decomposition algorithm: A literature review},
  author={Rahmaniani, Ragheb and Crainic, Teodor Gabriel and Gendreau, Michel and Rei, Walter},
  journal={European Journal of Operational Research},
  volume={259},
  number={3},
  pages={801--817},
  year={2017},
  publisher={Elsevier}
}

@article{crainic2021partial,
  title={Partial benders decomposition: general methodology and application to stochastic network design},
  author={Crainic, Teodor Gabriel and Hewitt, Mike and Maggioni, Francesca and Rei, Walter},
  journal={Transportation Science},
  volume={55},
  number={2},
  pages={414--435},
  year={2021},
  publisher={INFORMS}
}

@article{chen2020efficient,
  title={Efficient algorithm for finding k shortest paths based on re-optimization technique},
  author={Chen, Bi Yu and Chen, Xiao-Wei and Chen, Hui-Ping and Lam, William HK},
  journal={Transportation Research Part E: Logistics and Transportation Review},
  volume={133},
  pages={101819},
  year={2020},
  publisher={Elsevier}
}

@article{yen1971finding,
  title={Finding the k shortest loopless paths in a network},
  author={Yen, Jin Y},
  journal={management Science},
  volume={17},
  number={11},
  pages={712--716},
  year={1971},
  publisher={Informs}
}

@inproceedings{henderson2025automation,
  title={The Automation of Uncrewed Aircraft Systems Traffic Management Calibration Based on Experimental Platform Data},
  author={Henderson, Thomas C and Sacharny, David and Mello, Chad and Raley, William},
  booktitle={2025 IEEE International Conference on Robotics and Automation (ICRA)},
  pages={982--988},
  year={2025},
  organization={IEEE}
}

@article{bertsimas1998air,
  title={The air traffic flow management problem with enroute capacities},
  author={Bertsimas, Dimitris and Patterson, Sarah Stock},
  journal={Operations research},
  volume={46},
  number={3},
  pages={406--422},
  year={1998},
  publisher={INFORMS}
}

@article{lulli2007european,
  title={The European air traffic flow management problem},
  author={Lulli, Guglielmo and Odoni, Amedeo},
  journal={Transportation science},
  volume={41},
  number={4},
  pages={431--443},
  year={2007},
  publisher={INFORMS}
}

@article{guan2024exploration,
  title={The exploration and practice of low-altitude airspace flight service and traffic management in China},
  author={Guan, Xiangmin and Shi, Hongxia and Xu, Dongsong and Zhang, Binhua and Wei, Jian and Chen, Jun},
  journal={Green Energy and Intelligent Transportation},
  volume={3},
  number={2},
  pages={100149},
  year={2024},
  publisher={Elsevier}
}

@article{huang2024low,
  title={Low-altitude intelligent transportation: System architecture, infrastructure, and key technologies},
  author={Huang, Changqing and Fang, Shifeng and Wu, Hua and Wang, Yong and Yang, Yichen},
  journal={Journal of Industrial Information Integration},
  volume={42},
  pages={100694},
  year={2024},
  publisher={Elsevier}
}

@article{mohsan2023unmanned,
  title={Unmanned aerial vehicles (UAVs): Practical aspects, applications, open challenges, security issues, and future trends},
  author={Mohsan, Syed Agha Hassnain and Othman, Nawaf Qasem Hamood and Li, Yanlong and Alsharif, Mohammed H and Khan, Muhammad Asghar},
  journal={Intelligent service robotics},
  volume={16},
  number={1},
  pages={109--137},
  year={2023},
  publisher={Springer}
}

@article{mohsan2022towards,
  title={Towards the unmanned aerial vehicles (UAVs): A comprehensive review},
  author={Mohsan, Syed Agha Hassnain and Khan, Muhammad Asghar and Noor, Fazal and Ullah, Insaf and Alsharif, Mohammed H},
  journal={Drones},
  volume={6},
  number={6},
  pages={147},
  year={2022},
  publisher={Mdpi}
}

@article{yang2026comprehensive,
  title={A Comprehensive Review of Building the Resilience of Low-Altitude Logistics: Key Issues, Challenges, and Strategies.},
  author={Yang, Jingshuai and Xu, Haofeng},
  journal={Sustainability (2071-1050)},
  volume={18},
  number={1},
  year={2026}
}

@article{van2025stochastic,
  title={The stochastic dynamic postdisaster inventory allocation problem with trucks and UAVs},
  author={van Steenbergen, Robert M and van Heeswijk, Wouter JA and Mes, Martijn RK},
  journal={Transportation Science},
  volume={59},
  number={2},
  pages={360--390},
  year={2025},
  publisher={INFORMS}
}

@article{feng2025digital,
  title={Digital low-altitude airspace unmanned aerial vehicle path planning and operational capacity assessment in urban risk environments},
  author={Feng, Ouge and Zhang, Honghai and Tang, Weibin and Wang, Fei and Feng, Dikun and Zhong, Gang},
  journal={Drones},
  volume={9},
  number={5},
  pages={320},
  year={2025},
  publisher={MDPI}
}

@article{she2021efficiency,
  title={Efficiency of UAV-based last-mile delivery under congestion in low-altitude air},
  author={She, Ruifeng and Ouyang, Yanfeng},
  journal={Transportation Research Part C: Emerging Technologies},
  volume={122},
  pages={102878},
  year={2021},
  publisher={Elsevier}
}

@article{zhang2023research,
  title={Research on demand-based scheduling scheme of urban low-altitude logistics UAVs},
  author={Zhang, Honghai and Wu, Shixin and Feng, Ouge and Tian, Tian and Huang, Yuting and Zhong, Gang},
  journal={Applied Sciences},
  volume={13},
  number={9},
  pages={5370},
  year={2023},
  publisher={MDPI}
}

@article{jayaraman2003simulated,
  title={A simulated annealing methodology to distribution network design and management},
  author={Jayaraman, Vaidyanathan and Ross, Anthony},
  journal={European Journal of Operational Research},
  volume={144},
  number={3},
  pages={629--645},
  year={2003},
  publisher={Elsevier}
}

@article{venkatachalam2018two,
  title={A two-stage approach for routing multiple unmanned aerial vehicles with stochastic fuel consumption},
  author={Venkatachalam, Saravanan and Sundar, Kaarthik and Rathinam, Sivakumar},
  journal={Sensors},
  volume={18},
  number={11},
  pages={3756},
  year={2018},
  publisher={MDPI}
}

@article{faiz2024robust,
  title={A robust optimization framework for two-echelon vehicle and UAV routing for post-disaster humanitarian logistics operations},
  author={Faiz, Tasnim Ibn and Vogiatzis, Chrysafis and Liu, Jiongbai and Noor-E-Alam, Md},
  journal={Networks},
  volume={84},
  number={2},
  pages={200--219},
  year={2024},
  publisher={Wiley Online Library}
}

@article{bent2004scenario,
  title={Scenario-based planning for partially dynamic vehicle routing with stochastic customers},
  author={Bent, Russell W and Van Hentenryck, Pascal},
  journal={Operations research},
  volume={52},
  number={6},
  pages={977--987},
  year={2004},
  publisher={INFORMS}
}

@article{di2021trucks,
  title={Trucks and drones cooperation in the last-mile delivery process},
  author={Di Puglia Pugliese, Luigi and Macrina, Giusy and Guerriero, Francesca},
  journal={Networks},
  volume={78},
  number={4},
  pages={371--399},
  year={2021},
  publisher={Wiley Online Library}
}

@article{ibrahim2019bender,
  title={Bender's Decomposition for Optimization Design Problems in Communication Networks},
  author={Ibrahim, Ahmed and Dobre, Octavia A and Ngatched, Telex MN and Armada, Ana Garcia},
  journal={IEEE Network},
  volume={34},
  number={3},
  pages={232--239},
  year={2019},
  publisher={IEEE}
}

@article{stuive2024airspace,
  title={Airspace network design for urban UAV traffic management with congestion},
  author={Stuive, Leanne and Gzara, Fatma},
  journal={Transportation Research Part C: Emerging Technologies},
  volume={169},
  pages={104882},
  year={2024},
  publisher={Elsevier}
}

@article{grote2022sharing,
  title={Sharing airspace with uncrewed aerial vehicles (UAVs): Views of the general aviation (GA) community},
  author={Grote, Matt and Pilko, Aliaksei and Scanlan, James and Cherrett, Tom and Dickinson, Janet and Smith, Angela and Oakey, Andrew and Marsden, Greg},
  journal={Journal of Air Transport Management},
  volume={102},
  pages={102218},
  year={2022},
  publisher={Elsevier}
}

@article{he2024distributed,
  title={A distributed route network planning method with congestion pricing for drone delivery services in cities},
  author={He, Xinyu and Li, Lishuai and Mo, Yanfang and Huang, Jianxiang and Qin, S Joe},
  journal={Transportation Research Part C: Emerging Technologies},
  volume={160},
  pages={104536},
  year={2024},
  publisher={Elsevier}
}

@article{palmerius2024end,
  title={End-to-end drone route planning in flexible airspace design},
  author={Palmerius, Karljohan Lundin and Uggla, Alexander and Fylkner, Gustaf and Lundberg, Jonas},
  journal={Transportation Research Interdisciplinary Perspectives},
  volume={27},
  pages={101219},
  year={2024},
  publisher={Elsevier}
}

@article{guo2025planning,
  title={Planning UAM network under uncertain travelers’ preferences: A sequential two-layer stochastic optimization approach},
  author={Guo, Tao and Wu, Hao and Lu, Qing-Long and Antoniou, Constantinos},
  journal={Transportation Research Part A: Policy and Practice},
  volume={200},
  pages={104632},
  year={2025},
  publisher={Elsevier}
}

@article{chen2025hybrid,
  title={A hybrid centralized-decentralized traffic control framework for unmanned aerial vehicles in urban low-altitude airspace},
  author={Chen, Xiangdong and Li, Shen and Li, Meng},
  journal={Communications in Transportation Research},
  volume={5},
  pages={100195},
  year={2025},
  publisher={Elsevier}
}

@article{aarts2023capacity,
  title={Capacity of a constrained urban airspace: Influencing factors, analytical modelling and simulations},
  author={Aarts, Michiel JM and Ellerbroek, Joost and Knoop, Victor L},
  journal={Transportation Research Part C: Emerging Technologies},
  volume={152},
  pages={104173},
  year={2023},
  publisher={Elsevier}
}

@article{wang2023slot,
  title={Slot allocation for a multiple-airport system considering airspace capacity and flying time uncertainty},
  author={Wang, Yanjun and Liu, Chang and Wang, Hai and Duong, Vu},
  journal={Transportation Research Part C: Emerging Technologies},
  volume={153},
  pages={104185},
  year={2023},
  publisher={Elsevier}
}

@article{paul2025data,
  title={Data-driven optimization for drone delivery service planning with online demand},
  author={Paul, Aditya and Levin, Michael W and Waller, S Travis and Rey, David},
  journal={Transportation Research Part E: Logistics and Transportation Review},
  volume={198},
  pages={104095},
  year={2025},
  publisher={Elsevier}
}

@article{kishore2026scenario,
  title={Scenario-Based Platoon Lane Network Design},
  author={Kishore Bhoopalam, Anirudh and Agatz, Niels and Savelsbergh, Martin},
  journal={Networks},
  year={2026},
  publisher={Wiley Online Library}
}

@article{ramirez2023benders,
  title={Benders adaptive-cuts method for two-stage stochastic programs},
  author={Ram{\'\i}rez-Pico, Cristian and Ljubi{\'c}, Ivana and Moreno, Eduardo},
  journal={Transportation Science},
  volume={57},
  number={5},
  pages={1252--1275},
  year={2023},
  publisher={INFORMS}
}

@article{rebennack2016combining,
  title={Combining sampling-based and scenario-based nested Benders decomposition methods: application to stochastic dual dynamic programming},
  author={Rebennack, Steffen},
  journal={Mathematical Programming},
  volume={156},
  number={1},
  pages={343--389},
  year={2016},
  publisher={Springer}
}

@article{satici2026branch,
  title={A Branch-and-Benders Cut Algorithm for a Stochastic Service Network Design with Crowdsourced Capacity},
  author={Satici, Ozgur and Dayarian, Iman},
  journal={Transportation Science},
  year={2026},
  publisher={INFORMS}
}

@article{yu2021time,
  title={A time-consistent Benders decomposition method for multistage distributionally robust stochastic optimization with a scenario tree structure},
  author={Yu, Haodong and Sun, Jie and Wang, Yanjun},
  journal={Computational Optimization and Applications},
  volume={79},
  number={1},
  pages={67--99},
  year={2021},
  publisher={Springer}
}

@article{you2013multicut,
  title={Multicut Benders decomposition algorithm for process supply chain planning under uncertainty},
  author={You, Fengqi and Grossmann, Ignacio E},
  journal={Annals of Operations Research},
  volume={210},
  number={1},
  pages={191--211},
  year={2013},
  publisher={Springer}
}

@article{rajan2022routing,
  title={Routing problem for unmanned aerial vehicle patrolling missions-a progressive hedging algorithm},
  author={Rajan, Sudarshan and Sundar, Kaarthik and Gautam, Natarajan},
  journal={Computers \& Operations Research},
  volume={142},
  pages={105702},
  year={2022},
  publisher={Elsevier}
}

@article{liu2023multi,
  title={Multi-period stochastic programming for relief delivery considering evolving transportation network and temporary facility relocation/closure},
  author={Liu, Kanglin and Yang, Liu and Zhao, Yejia and Zhang, Zhi-Hai},
  journal={Transportation research part e: logistics and transportation review},
  volume={180},
  pages={103357},
  year={2023},
  publisher={Elsevier}
}

@article{lin2021branch,
  title={Branch-and-cut approach based on generalized benders decomposition for facility location with limited choice rule},
  author={Lin, Yun Hui and Tian, Qingyun},
  journal={European Journal of Operational Research},
  volume={293},
  number={1},
  pages={109--119},
  year={2021},
  publisher={Elsevier}
}

@article{papadakos2008practical,
  title={Practical enhancements to the Magnanti--Wong method},
  author={Papadakos, Nikolaos},
  journal={Operations Research Letters}, volume={36}, number={4}, pages={444--449},
  year={2008}, publisher={Elsevier}
}

@article{crainic2000service,
  author  = {Crainic, Teodor Gabriel},
  title   = {Service network design in freight transportation},
  journal = {European Journal of Operational Research},
  volume  = {122},
  number  = {2},
  pages   = {272--288},
  year    = {2000},
  doi     = {10.1016/S0377-2217(99)00233-7}
}

@article{bai2014stochastic,
  author  = {Bai, Ruibin and Wallace, Stein W. and Li, Jingpeng and Chong, Alain Yee-Loong},
  title   = {Stochastic service network design with rerouting},
  journal = {Transportation Research Part B: Methodological},
  volume  = {60},
  pages   = {50--65},
  year    = {2014},
  doi     = {10.1016/j.trb.2013.11.001}
}

@article{wang2019stochastic,
  author  = {Wang, Xin and Crainic, Teodor Gabriel and Wallace, Stein W.},
  title   = {Stochastic network design for planning scheduled transportation services: The value of deterministic solutions},
  journal = {INFORMS Journal on Computing},
  volume  = {31},
  number  = {1},
  pages   = {153--170},
  year    = {2019},
  doi     = {10.1287/ijoc.2018.0819}
}

@article{hewitt2022scheduled,
  author  = {Hewitt, Mike},
  title   = {The flexible scheduled service network design problem},
  journal = {Transportation Science},
  volume  = {56},
  number  = {4},
  pages   = {1000--1021},
  year    = {2022},
  doi     = {10.1287/trsc.2022.1129}
}

@article{crainic2014progressive,
  author  = {Crainic, Teodor Gabriel and Hewitt, Mike and Rei, Walter},
  title   = {Scenario grouping in a progressive hedging-based meta-heuristic for stochastic network design},
  journal = {Computers \& Operations Research},
  volume  = {43},
  pages   = {90--99},
  year    = {2014},
  doi     = {10.1016/j.cor.2013.08.020}
}

@article{moradi2025systematic,
  title={A systematic review of sustainable ground-based last-mile delivery of parcels: Insights from operations research},
  author={Moradi, Nima and Mafakheri, Fereshteh and Wang, Chun},
  journal={Vehicles},
  volume={7},
  number={4},
  pages={121},
  year={2025},
  publisher={MDPI}
}

@article{ford1956maximal,
  author  = {Ford, L. R. and Fulkerson, D. R.},
  title   = {Maximal flow through a network},
  journal = {Canadian Journal of Mathematics},
  volume  = {8},
  pages   = {399--404},
  year    = {1956},
  doi     = {10.4153/CJM-1956-045-5}
}

@article{nawaz2026towards,
  title={Towards regulating human oversight: challenges for EU drone law},
  author={Nawaz, Samar Abbas},
  journal={Information \& Communications Technology Law},
  volume={35},
  number={2},
  pages={234--250},
  year={2026},
  publisher={Taylor \& Francis}
}

\appendix

\section{Proofs}\label{app:proofs}

\REFTHR{\begin{proof}[Proof of Proposition~\ref{prop:equivalence}]
Fix $\omega$ and $\boldsymbol{x}$. Since every arc of the time-expanded network strictly increases the time index (flight arcs by $\tau_\ell\ge 1$ periods), the network is a finite directed acyclic graph (DAG); in particular it contains no directed cycle. We construct $\Psi$ and a left inverse $\Phi$ explicitly and verify well-definedness, feasibility, and objective preservation.

\textit{Step 1 (construction of $\Psi$: paths expand into arc flows).}
Let $\boldsymbol{y}\in\mathcal{F}_P(\boldsymbol{x};\omega)$. For each request $k$ with $\sum_{p\in\mathcal{P}_k} y_{kp}^\omega=1$, let $p$ denote its selected path; set $f_{k,\ell,t_0}^\omega=1$ exactly on the flight-arc departures used by $p$, set $z_k^\omega=1$, and set $h_{kt}^\omega=1$ at the arrival time embedded in $p$; set all variables of unassigned requests to zero. Conditions (i)--(iii) of Definition~1 imply that this $(\boldsymbol{f},\boldsymbol{z},\boldsymbol{h})$ satisfies \eqref{eq:arc_source_l}--\eqref{eq:arc_sink_l}. Because $\boldsymbol f$ is constructed to carry exactly the occupancy of the selected paths, the identity \eqref{eq:occupancy_identity} below applies directly to $\boldsymbol f$ itself, so $\boldsymbol{y}$ satisfying \eqref{eq:cap} implies $\boldsymbol f$ satisfies \eqref{eq:arc_cap_l}; the objective is preserved by construction. Hence $\Psi(\boldsymbol y):=(\boldsymbol f,\boldsymbol z,\boldsymbol h)\in\mathcal F_A(\boldsymbol x;\omega)$. Distinct $\boldsymbol{y},\boldsymbol{y}'\in\mathcal{F}_P(\boldsymbol{x};\omega)$ differ in the path selected for at least one request, hence activate at least one different flight-arc departure, so $\Psi$ is injective.

\textit{Step 2 (construction of $\Phi$: arc flows decompose into paths).}
Conversely, let $(\boldsymbol{f},\boldsymbol{z},\boldsymbol{h})\in\mathcal{F}_A(\boldsymbol{x};\omega)$ and fix a request $k\in K^\omega$. Constraints \eqref{eq:arc_source_l}--\eqref{eq:arc_sink_l} pin down the divergence of $(f_{k,\ell,t}^\omega)_{\ell,t}$ at $(o_k,t_k)$ and at every candidate destination node $(d_k,t)$, but --- unlike every other node of the time-expanded network --- do not restrict it at $(o_k,t)$ for $t\neq t_k$; the only requirement there is the capacity constraint~\eqref{eq:arc_cap_l}.
If $z_k^\omega=0$, then \eqref{eq:arc_arrive_once_l} forces $h_{kt}^\omega=0$ for all $t$, so $(f_{k,\cdot,\cdot}^\omega)$ has zero divergence at $(o_k,t_k)$ and at every destination copy; by \eqref{eq:arc_flow_l} it also has zero divergence at every other node except possibly the time-copies of $o_k$ with $t\neq t_k$, so $(f_{k,\cdot,\cdot}^\omega)$ decomposes, by the flow-decomposition theorem \citep{ford1956maximal}, into unit cycle flows confined to $\{(o_k,t):t\neq t_k\}$; a DAG admits no nonzero circulation, so in fact $f_{k,\ell,t}^\omega=0$ for all $(\ell,t)$ --- a rejected request consumes no corridor--time capacity.
If $z_k^\omega=1$, then \eqref{eq:arc_arrive_once_l} activates exactly one destination copy $(d_k,t_k^*)$, so $(f_{k,\cdot,\cdot}^\omega)$ has divergence $+1$ at $(o_k,t_k)$, divergence $-1$ at $(d_k,t_k^*)$, divergence $0$ at every node other than the time-copies of $o_k$, and unrestricted (but capacity-bounded) divergence at $(o_k,t)$ for $t\neq t_k$. By the flow-decomposition theorem, $(f_{k,\cdot,\cdot}^\omega)$ decomposes into unit path flows and unit cycle flows; every cycle must be confined to nodes of zero divergence, hence to $\{(o_k,t):t\neq t_k\}$, which a DAG cannot support, so there are no cycle components. Exactly one path component carries the unit divergence from $(o_k,t_k)$ to $(d_k,t_k^*)$; call it $p_k$. Any remaining path components necessarily begin and end within $\{(o_k,t):t\neq t_k\}$ (the only nodes with divergence not already accounted for by $p_k$), since $p_k$ already saturates the divergence at $(o_k,t_k)$ and $(d_k,t_k^*)$ and every other node is balanced. We define $\Phi$ by keeping only $p_k$: set $f_{k,\ell,t_0}^\omega=1$ exactly on the flight-arc departures used by $p_k$ and discard any remaining components. This discarding changes neither $z_k^\omega$ nor $h_{kt}^\omega$ (which the remaining components, having no endpoint at $(o_k,t_k)$ or a destination copy, do not touch) and can only relax the capacity constraint~\eqref{eq:arc_cap_l}. Conditions (i)--(iii) of Definition~1 hold for $p_k$ by construction, so $p_k\in\mathcal{P}_k$.
Define $\Phi(\boldsymbol{f},\boldsymbol{z},\boldsymbol{h})=\boldsymbol{y}$ by $y_{kp_k}^\omega=1$ for every accepted request $k$ and $y_{kp}^\omega=0$ otherwise; note $\Phi$ need not be injective, since distinct $(\boldsymbol f,\boldsymbol z,\boldsymbol h)$ differing only in the discarded components above share the same image. Constraint \eqref{eq:link_eq} holds because each request selects at most one path. For the capacity constraints, the definition \eqref{eq:a_indicator} of the path occupancy yields, for every corridor--time pair $(\ell,t)$, the identity
\begin{align}
\sum_{k\in K^\omega}\sum_{p\in\mathcal{P}_k} a_{kp}^{\ell t}\, y_{kp}^{\omega}
\;=\;\sum_{k\in K^\omega}\sum_{t_0\in\mathcal{T}} b_{\ell,t_0}^{\ell t}\, f_{k,\ell,t_0}^{\omega},
\label{eq:occupancy_identity}
\end{align}
whenever $\boldsymbol f$ carries exactly the occupancy of the paths $\{p_k\}_{k}$ selected by $\boldsymbol y$: on the left-hand side only the selected paths $p_k$ contribute, on the right-hand side only the departures used by those paths carry flow, and rejected requests contribute zero on both sides. In Step~2, $\boldsymbol f$ may additionally carry the discarded components identified above; since these are nonnegative, the true occupancy of $\boldsymbol f$ on the right-hand side of \eqref{eq:arc_cap_l} is at least the right-hand side of \eqref{eq:occupancy_identity}, so $\boldsymbol{f}$ satisfying \eqref{eq:arc_cap_l} implies $\boldsymbol y$ satisfies \eqref{eq:cap}. Finally, $\sum_{k\in K^\omega} r_k z_k^\omega=\sum_{k\in K^\omega}\sum_{p\in\mathcal{P}_k} r_k y_{kp}^\omega$, so $\Phi$ maps $\mathcal{F}_A(\boldsymbol{x};\omega)$ into $\mathcal{F}_P(\boldsymbol{x};\omega)$ and preserves the objective.

\textit{Step 3 ($\Phi$ is a left inverse of $\Psi$).}
Fix $\boldsymbol{y}\in\mathcal{F}_P(\boldsymbol{x};\omega)$ and let $(\boldsymbol f,\boldsymbol z,\boldsymbol h)=\Psi(\boldsymbol y)$. By construction, $(f_{k,\cdot,\cdot}^\omega)$ has no divergence at any time-copy of $o_k$ other than $(o_k,t_k)$, so the decomposition in Step~2 has no components to discard, and the single path component recovered there is exactly the path selected by $\boldsymbol y$. Hence $\Phi(\Psi(\boldsymbol y))=\boldsymbol y$, i.e., $\Phi\circ\Psi=\mathrm{id}_{\mathcal{F}_P(\boldsymbol{x};\omega)}$. Since $\Psi$ preserves the objective on all of $\mathcal F_P(\boldsymbol x;\omega)$ (Step~1) and $\Phi$ preserves the objective on all of $\mathcal F_A(\boldsymbol x;\omega)$ (Step~2), taking maxima gives $Q(\boldsymbol{x};\omega)\ge G(\boldsymbol{x};\omega)$ --- apply $\Psi$ to an optimal element of $\mathcal F_P(\boldsymbol x;\omega)$ --- and $Q(\boldsymbol{x};\omega)\le G(\boldsymbol{x};\omega)$ --- apply $\Phi$ to an optimal element of $\mathcal F_A(\boldsymbol x;\omega)$ --- so $Q(\boldsymbol{x};\omega)=G(\boldsymbol{x};\omega)$. Moreover, for any optimal $\boldsymbol{y}^*\in\mathcal{F}_P(\boldsymbol{x};\omega)$, $\Psi(\boldsymbol{y}^*)$ is an optimal element of $\mathcal{F}_A(\boldsymbol{x};\omega)$ with the same objective value, establishing the claimed correspondence.
$\hfill\blacksquare$
\end{proof}}

\REVII{\begin{proof}[Proof of Corollary~\ref{cor:lp_equivalence}]
The fractional analogues of Steps~1--2 above hold by flow decomposition on the acyclic time-expanded network \citep{ford1956maximal}: every feasible fractional arc flow of request $k$ with acceptance level $z_k^\omega\in[0,1]$ decomposes into a nonnegative combination of unit path flows of total weight $z_k^\omega$, and every fractional path solution aggregates arc-wise into a feasible arc flow. Both maps preserve the objective and, through the identity \eqref{eq:occupancy_identity}, the capacity consumption at every $(\ell,t)$. Hence the two LP value functions coincide at \textit{every} right-hand side $\boldsymbol{x}\ge 0$, i.e., they are identical as functions of $\boldsymbol{x}$. Since both value functions are concave and piecewise linear in $\boldsymbol{x}$, and the optimal capacity duals of a right-hand-side-parameterized LP are exactly the supergradients of its value function at $\boldsymbol{x}$, the two dual sets coincide at every $\boldsymbol{x}$ as well.
$\hfill\blacksquare$
\end{proof}}

\begin{proof}[Validity of the optimality cut~\eqref{eq:opt_cut}]
The validity of~\eqref{eq:opt_cut} follows from linear programming duality.
For a fixed scenario $\omega$, introduce dual variables
$\alpha_k^\omega\ge 0$ for the assignment constraints~\eqref{eq:sub_assign},
$\lambda_{\ell t}^\omega\ge 0$ for the capacity constraints~\eqref{eq:sub_cap},
and $\gamma_{kp}^\omega\ge 0$ for the upper-bound constraints in~\eqref{eq:sub_box}.
The Lagrangian of the subproblem is
\begin{align}
\mathcal{L}(y^\omega;\alpha^\omega,\lambda^\omega,\gamma^\omega)=&\sum_{k\in K^\omega}\sum_{p\in\mathcal{P}_k} r_k y_{kp}^\omega+ \sum_{k\in K^\omega} \alpha_k^\omega\Big(1-\sum_{p\in\mathcal{P}_k} y_{kp}^\omega\Big) \notag\\
&\quad + \sum_{(\ell,t)}\lambda_{\ell t}^\omega\Big(x_{\ell t}-\sum_{k\in K^\omega}\sum_{p\in\mathcal{P}_k} a_{kp}^{\ell t}y_{kp}^\omega\Big)+ \sum_{k\in K^\omega}\sum_{p\in\mathcal{P}_k}\gamma_{kp}^\omega(1-y_{kp}^\omega) \notag\\
=&\sum_{k\in K^\omega}\alpha_k^\omega+ \sum_{(\ell,t)}\lambda_{\ell t}^\omega x_{\ell t}+ \sum_{k\in K^\omega}\sum_{p\in\mathcal{P}_k}\gamma_{kp}^\omega \notag\\
&\quad + \sum_{k\in K^\omega}\sum_{p\in\mathcal{P}_k}\Big(r_k-\alpha_k^\omega-\sum_{(\ell,t)}\lambda_{\ell t}^\omega a_{kp}^{\ell t}-\gamma_{kp}^\omega\Big)y_{kp}^\omega .
\label{eq:lagrangian}
\end{align}

Because $y_{kp}^\omega\ge 0$, the supremum over $y^\omega$ is finite only if
\begin{align}
r_k-\alpha_k^\omega-\sum_{(\ell,t)}\lambda_{\ell t}^\omega a_{kp}^{\ell t}-\gamma_{kp}^\omega \le 0,
\qquad \forall k\in K^\omega,\; p\in\mathcal{P}_k.
\label{eq:dual_feas}
\end{align}
Under these conditions, the dual function is
\begin{align}
g(\alpha^\omega,\lambda^\omega,\gamma^\omega;x)=\sum_{k\in K^\omega}\alpha_k^\omega+ \sum_{k\in K^\omega}\sum_{p\in\mathcal{P}_k}\gamma_{kp}^\omega+ \sum_{(\ell,t)}\lambda_{\ell t}^\omega x_{\ell t}.
\label{eq:dual_func}
\end{align}

Weak duality implies that for any dual-feasible multipliers,
\begin{align}
G_{\mathrm{LP}}(x;\omega)\leq \sum_{(\ell,t)}\lambda_{\ell t}^\omega x_{\ell t}+ \Big(\sum_{k\in K^\omega}\alpha_k^\omega+ \sum_{k\in K^\omega}\sum_{p\in\mathcal{P}_k}\gamma_{kp}^\omega\Big).
\label{eq:weak_duality_bound}
\end{align}
Strong duality ensures equality at $x^n$, which yields the supporting hyperplane~\eqref{eq:opt_cut}.
$\hfill \blacksquare$
\end{proof}

\REVII{\textbf{Validity of the \DBTP certificate.} The multi-fidelity certificate reported in Section~\ref{sec:exp_heuristics} is obtained by the same repair as above, applied at termination over the full path set rather than mid-iteration over the truncated pool $\widehat{\mathcal P}_k$. Fix $\boldsymbol x$ and $\omega$, and let $(\alpha^\omega,\lambda^\omega)$ be the optimal dual multipliers of the restricted master LP that \DBTP solves over $\widehat{\mathcal P}_k$. One additional pricing pass over the full path set $\mathcal P_k$ identifies the maximum reduced-cost violation $\bar r_k(\lambda^\omega)=r_k-\alpha_k^\omega-\min_{p\in\mathcal P_k}\sum_{(\ell,t)\in p}\lambda_{\ell t}^\omega$; setting $\alpha_k^{\omega,\mathrm{rep}}=\alpha_k^\omega+\max\{0,\bar r_k(\lambda^\omega)\}$ restores the dual-feasibility condition~\eqref{eq:dual_feas} for every $p\in\mathcal P_k$, exactly as in the dual-feasible cut repair of Section~\ref{sec:stabilization}, so $(\lambda^\omega,\alpha^{\omega,\mathrm{rep}})$ is dual-feasible for the full LP-relaxed subproblem~\eqref{eq:sub_obj}--\eqref{eq:sub_box}. Weak duality~\eqref{eq:weak_duality_bound} then gives
\begin{align}
G_{\mathrm{LP}}(\boldsymbol x;\omega)\le \sum_{(\ell,t)}\lambda_{\ell t}^\omega x_{\ell t}+\sum_{k\in K^\omega}\alpha_k^{\omega,\mathrm{rep}}.
\end{align}
Aggregating this bound over scenarios and subtracting the reservation cost yields a valid upper bound $\mathrm{UB}$ on the true (non-truncated) optimal objective at $\boldsymbol x$, which is exactly the certificate underlying the certificate gap $(\mathrm{UB}-z_{\mathrm{DBTP}})/\mathrm{UB}$ reported in Section~\ref{sec:exp_heuristics}.}

\REVII{\begin{remark}[Cut validity of the stabilized \DBCG procedure]
\label{prop:convergence}
Consider Algorithm~\ref{alg:multicut_benders} applied to the LP-relaxed
two-stage formulation augmented with the stabilization mechanisms described in Section~\ref{sec:stabilization}. Suppose that, for every scenario $\omega$ and queried capacity vector $\boldsymbol{x}$, the subproblem solver returns the exact LP recourse value $G_{\mathrm{LP}}(\boldsymbol{x};\omega)$ together with a dual-feasible multiplier vector --- a condition satisfied by the arc-flow LP solver of Section~\ref{sec:subproblem_arcflow} and by the column-generation solver of Section~\ref{sec:subproblem_cg} equipped with the dual-feasibility repair of Section~\ref{sec:stabilization}. Then, by \eqref{eq:weak_duality_bound}, every cut added to the master is a valid supporting hyperplane of the (concave) recourse value function, regardless of whether it is generated at the current master solution $x^n$ or at the Papadakos auxiliary point $x^0$, which is likewise master-feasible. We do not claim, and do not need, that this iterative process converges in finitely many steps: cut validity alone is what makes $\mathrm{gap}_{\mathrm{Bend}}$ a meaningful bound on the true optimality gap at whatever iteration the algorithm is stopped.
$\hfill\blacksquare$
\end{remark}}

\REVII{\section{Path-Enumeration Details}\label{app:pathgen}}

\REFONE{The deviation step at the spur node $v_i$ has two effects that together ensure each newly enumerated candidate is genuinely distinct from previously accepted paths. First, removing the nodes $\{v_0,\ldots,v_{i-1}\}$ from the deviation graph keeps the spur portion simple and disjoint from the root interior. Second, and crucially, removing for every previously accepted path that shares the same root $R_i$ the outgoing edge $(v_i,v_{i+1}^{(q')})$ used by that path at $v_i$ forces the spur to leave $v_i$ on a strictly different edge; without this edge removal the spur from $v_i$ would simply recompute the suffix of $p^{(q-1)}$, returning a duplicate of an already-known path. The corrected pseudocode in Algorithm~\ref{alg:p_short} now lists both operations explicitly. Our implementation uses \texttt{networkx.shortest\_simple\_paths} (which internally realizes Yen's algorithm with the same edge-removal logic) wrapped by the time-feasibility filter $s_k+\sum_{\ell\in c}\tau_\ell\le\bar s_k$ and the dominance pruning above; we have re-verified on the full benchmark set that for every request $k$ the realized pool $\mathcal P_k$ contains $P$ \textit{distinct} simple paths whenever the network admits that many, so the computational results in Sections~\ref{sec:numerical study}--\ref{sec:case study} are unaffected by the earlier typographical omission.}

\begin{algorithm}[t]
\caption{\REFONE{$P$-Shortest Feasible Simple Paths (Yen's deviation scheme)}}
\label{alg:p_short}
\KwIn{$G=(\mathcal{V},\mathcal L)$; request $k$ with $(o_k,d_k,s_k,\bar s_k)$; limit $P$}
\KwOut{$\mathcal P_k$}

\textbf{Notation:} $s_k$ denotes the departure time of request $k$; a path is simple if it contains no repeated nodes;
concatenation joins two path segments end-to-end;
the travel time of a path $p$ is $\sum_{\ell\in p}\tau_\ell$.\;

Compute the shortest feasible simple path $p^{(0)}$ from $o_k$ to $d_k$\;
Set $\mathcal P_k\leftarrow\{p^{(0)}\}$ and a candidate heap $\mathcal C\leftarrow\emptyset$\;

\For{$q=1$ \KwTo $P-1$}{
    Let the previously accepted path be $p^{(q-1)}=(v_0,\ldots,v_M)$\;
    \For(\REFONE{\tcp*[h]{spur node loop}}){$i=0$ \KwTo $M-1$}{
        \REFONE{Let the \textit{root path} be $R_i=(v_0,\ldots,v_i)$ and the \textit{spur node} be $v_i$\;
        Construct the deviation graph $G_i$ from $G$ by:
        (a) removing every node in $\{v_0,\ldots,v_{i-1}\}$ (and all incident edges) to keep the spur path simple and disjoint from the root interior; and
        (b) for every previously accepted path $p^{(q')} \in\mathcal P_k\cup\{p^{(q-1)}\}$ that shares the same root $R_i$, removing the outgoing edge $(v_i, v_{i+1}^{(q')})$ used by $p^{(q')}$ at the spur node\;
        Compute a shortest simple path $\tilde p$ from $v_i$ to $d_k$ in $G_i$\;}
        \If{$\tilde p$ exists}{
            Let $c$ be the concatenation of $R_i$ and $\tilde p$\;
            \If{$s_k+\sum_{\ell\in c}\tau_\ell \le \bar s_k$ \REFONE{and $c\notin \mathcal P_k\cup\mathcal C$}}{
                insert $c$ into $\mathcal C$ ordered by total travel time
            }
        }
    }
    \If{$\mathcal C=\emptyset$}{break}
    Extract shortest path in $\mathcal C$ and append to $\mathcal P_k$
}
\Return{$\mathcal P_k$}
\end{algorithm}

\rev{\subsection{Direct Arc-Flow LP Subproblem}
\label{sec:subproblem_arcflow}

The column-generation procedure in Section~\ref{sec:subproblem_cg} represents each near-parallel route as a distinct column and therefore \textit{pays for} the spatial degeneracy of the recourse LP (Section~\ref{sec:stabilization}): its restricted-master cost inflates with the persistent column pool. The arc-flow formulation of Section~\ref{subsec:arc_based} represents all parallel routes \textit{compactly} on shared arcs and is therefore structurally \textit{immune} to spatial degeneracy, at the cost of a larger single LP. This opposite structural response --- and not merely a computational trade-off --- makes the arc-flow subproblem a natural alternative to column generation in regimes where the column pool would otherwise dominate the computational cost. We refer to this second solver as the \textit{arc-flow subproblem}.

\textbf{Formulation.}
Drop the integrality requirements in~\eqref{eq:arc_stage2_obj_l}--\eqref{eq:arc_sink_l} and introduce continuous variables $f_{k,\ell,t_0}^\omega\ge 0$, $z_k^\omega\in[0,1]$.  The resulting LP is
\begin{align}
G_{\mathrm{LP}}^{\mathrm{AF}}(\boldsymbol{x}^n;\omega)=\max\quad
&\sum_{k\in K^\omega} r_k\,z_k^\omega
\label{eq:arcflow_obj}\\
\text{s.t.}\quad
&\sum_{k\in K^\omega}\sum_{t_0\in\mathcal{T}}
  b_{\ell,t_0}^{\ell t}\,f_{k,\ell,t_0}^\omega
\le x^n_{\ell t},
\quad\forall(\ell,t),
\label{eq:arcflow_cap}\\
&\text{flow-conservation constraints~\eqref{eq:arc_source_l}--\eqref{eq:arc_sink_l}
  with continuous variables,}
\notag\\
&f_{k,\ell,t_0}^\omega\ge 0,\quad 0\le z_k^\omega\le 1.
\label{eq:arcflow_bounds}
\end{align}
The LP has $O\!\bigl(|K^\omega|\cdot(|\mathcal{L}|\cdot|\mathcal{T}|+|\mathcal{V}|\cdot|\mathcal{T}|)\bigr)$ variables and
$O\!\bigl(|\mathcal{L}|\cdot|\mathcal{T}|+|K^\omega|\cdot|\mathcal{V}|\cdot|\mathcal{T}|\bigr)$ constraints — polynomial and fixed for a given instance, with no path enumeration required.

\REVII{\textbf{Equivalence.}
By Corollary~\ref{cor:lp_equivalence}, $G_{\mathrm{LP}}^{\mathrm{AF}}(\boldsymbol{x}^n;\omega)=G_{\mathrm{LP}}(\boldsymbol{x}^n;\omega)$ for every $\boldsymbol{x}^n$, and the dual $\lambda_{\ell t}^{\omega*}$ of constraint~\eqref{eq:arcflow_cap} lies in the same optimal dual set as the capacity dual of the path-packing LP~\eqref{eq:sub_cap}. The Benders cut
\begin{align}
\theta_\omega\;\le\;\bigl(\lambda^{\omega*}\bigr)^\top x
+G_{\mathrm{LP}}^{\mathrm{AF}}(\boldsymbol{x}^n;\omega)
-\bigl(\lambda^{\omega*}\bigr)^\top\boldsymbol{x}^n
\label{eq:arcflow_cut}
\end{align}
built from the arc-flow solution is not necessarily identical to the cut~\eqref{eq:opt_cut} built from the path-packing solution, but both are supporting hyperplanes of the same value function $G_{\mathrm{LP}}(\cdot;\omega)=G_{\mathrm{LP}}^{\mathrm{AF}}(\cdot;\omega)$ and are equally valid for the same master problem~\eqref{eq:master_obj}--\eqref{eq:master_cuts}.}

\textbf{Computational properties.}
Compared with the column-generation approach, the arc-flow subproblem offers three advantages.
First, the LP is solved in one pass, yielding \textit{exact} dual multipliers without a pricing loop.
Second, there is no column pool to maintain: the LP size is constant across Benders iterations, preventing the unbounded pool growth that afflicts CG-based implementations on instances with many near-optimal paths (i.e., where the LP dual has many columns with near-zero reduced cost).
Third, warm-starting across Benders iterations is straightforward: only the right-hand sides $x^n_{\ell t}$ of constraints~\eqref{eq:arcflow_cap} change between iterations.
\REVII{Because the path-packing recourse is strongly degenerate (both primal and dual; Section~\ref{sec:stabilization}), the arc-flow LP is best solved by the \textit{barrier method without crossover} rather than by simplex: stopping at an interior point avoids the many degenerate pivots that make simplex slow on this structure, and returns \textit{central} (analytic-centre-like) dual multipliers that yield deeper, better-conditioned Benders cuts than the vertex duals of simplex. The trade-off is the loss of basis warm-starting across Benders iterations.
The results reported in Section~\ref{sec:numerical study} use column generation throughout; the arc-flow subproblem is presented here as a theoretically-equivalent, structurally motivated alternative for settings where the column pool would otherwise dominate the computational cost.}}

\section{Additional Computational Diagnostics}
\label{app:diagnostics}

This appendix reports five supplementary items that support the implementation choices and modeling assumptions used in the main text. Specifically, we give (i) the full experimental setup and instance generator, (ii) the full discussion of connectivity-induced path degeneracy summarized in Section~\ref{sec:path_degeneracy}, (iii) the empirical second-stage integrality gap, and (iv) the comparison of \DBCG against \ABCG under a fixed computational budget.

\subsection{Experimental Setup and Instance Generation}
\label{app:setup}

\textbf{Experimental setup.} All algorithms are implemented in Python~3.11 with \texttt{NumPy}~1.26, \texttt{NetworkX}~3.2, and \texttt{Numba}~0.59 (used to JIT-compile the pricing sweep in Algorithm~\ref{alg:cg_subproblem}). All linear programs (the Benders master, the path-packing restricted master, and the post-optimization integer recovery in Section~\ref{sec:integrality}) are solved with Gurobi~11.0 via \texttt{gurobipy}; the master and the path-packing restricted master use the dual simplex method with default presolve and basis warm-starting between successive solves. Experiments are run on a workstation with a 16-core Intel Xeon W-2245 CPU at 3.90\,GHz, 128\,GB DDR4 ECC RAM, and Ubuntu~22.04 LTS; Gurobi is configured to use a single thread per LP solve to make per-instance wall-clock times comparable across runs. The master-LP tolerance is $\varepsilon=10^{-4}$ and the column-generation tolerance is $\varepsilon_{\mathrm{cg}}=10^{-6}$; each reported entry is the mean over independent demand seeds (up to five). The performance tables report Benders iteration counts and total wall-clock time.

\textbf{Instance generation.} Unless otherwise stated, experiments are conducted on randomly generated corridor networks. Nodes are placed on a jittered Cartesian lattice and connected to their four nearest neighbours by bidirectional corridors; the spatial perturbation creates heterogeneous corridor lengths and travel times. For each scenario $\omega$, requests are generated by sampling origins, destinations, and release times uniformly, while request revenues are set proportional to shortest-path distances with small random perturbations to avoid ties. Corridor capacities are generated according to the expected shortest-path demand, and reservation costs are uniform, $c_{\ell t}=c_0$. Three parameters characterize the congestion regime: route overlap $\eta_{\mathrm{ov}}$, capacity tightness $\tau_{\mathrm{cap}}$, and release-window ratio $W$; the baseline setting used throughout, unless a specific experiment states otherwise, is $\eta_{\mathrm{ov}}=0.70$, $\tau_{\mathrm{cap}}=0.55$, $W=0.25$, and detour ratio $\rho=2.0$. Scaling experiments increase network size, horizon, demand volume, scenario count, and path-pool size jointly according to Table~\ref{tab:exp1_combined}, whereas structural experiments vary only $\eta_{\mathrm{ov}}$.

\subsection{Path Degeneracy: Further Discussion and Distinction from Classical VRP Degeneracy}
\label{app:path_degeneracy}

Section~\ref{sec:path_degeneracy} summarizes how network connectivity governs path degeneracy in the CG recourse. Two points deserve fuller discussion.

\paragraph{Connectivity versus network size.} The degeneracy phenomenon is primarily determined by network connectivity rather than network size. In the regular lattice instances used throughout this paper, larger instances are also more densely connected, so connectivity increases together with $|\mathcal V|$. For more general corridor networks, however, two instances with the same number of nodes may exhibit substantially different CG performance because their routing alternatives differ.

\paragraph{Distinction from classical VRP degeneracy.} The degeneracy considered here differs fundamentally from that encountered in classical vehicle-routing formulations. In VRP models, computational difficulty is typically caused by combinatorial symmetry, such as interchangeable vehicles or equivalent customer visit orders, which is commonly addressed through symmetry-breaking constraints or branch-and-bound strategies. In contrast, the present formulation contains neither vehicle-index symmetry nor route-order symmetry. Instead, degeneracy arises because the corridor network admits many near-equivalent feasible paths with similar reduced costs. The resulting difficulty is therefore geometric rather than combinatorial, manifesting itself as LP degeneracy in the restricted master instead of search tree symmetry. Consequently, improving computational performance depends primarily on controlling path generation rather than on symmetry-breaking techniques.

\subsection{Second-stage Integrality Check}
\label{sec:app_integrality}

\REFTHR{Although the second-stage recourse is solved through its LP relaxation within the Benders framework, Proposition~\ref{prop:equivalence} does not theoretically guarantee a zero integrality gap. We therefore verify this property numerically by comparing the optimal LP and MIP recourse values under the rounded first-stage solution $\boldsymbol{x}^*=\lceil\boldsymbol{x}_{\mathrm{LP}}^*\rceil$. Across all tested instances,
$$\max_{(|\mathcal V|,\mathrm{seed},\omega)}
\left|G_{\mathrm{LP}}(\boldsymbol{x}^*;\omega)-
G_{\mathrm{MIP}}(\boldsymbol{x}^*;\omega)\right|\le 10^{-6},
$$
confirming that the LP relaxation and the corresponding MIP attain identical objective values up to numerical precision.}

\begin{table}[!htbp]
\centering
\caption{Natural integrality of the second-stage LP relaxation before tie-breaking.}
\label{tab:integrality}
\begin{tabular}{cccc}
\toprule
$|\mathcal V|$ & Integral / Total & Percentage (\%) & Max.\ objective difference \\
\midrule
8  & 22 / 25 & 88.0 & $\le 10^{-6}$ \\
10 & 21 / 25 & 84.0 & $\le 10^{-6}$ \\
12 & 21 / 25 & 84.0 & $\le 10^{-6}$ \\
\midrule
Total & 64 / 75 & 85.3 & $\le 10^{-6}$ \\
\bottomrule
\end{tabular}
\end{table}

\REFTHR{Table~\ref{tab:integrality} summarizes the natural integrality of the second-stage LP relaxation before applying the tie-breaking objective. Overall, 64 of the 75 tested instances (85.3\%) produce naturally integral LP solutions, with similar proportions across all network sizes (84\%--88\%). The remaining 11 instances exhibit fractional LP solutions due to degeneracy rather than numerical error. In particular, ten instances have a maximum deviation from integrality of exactly 0.5, corresponding to an equal split between two equivalent optimal routes, while one instance has a maximum deviation of 0.333, corresponding to three equivalent routes. These observations explain the role of the optional tie-breaking objective in selecting one representative optimal solution from multiple equivalent LP optima, while confirming that the zero objective gap reported above remains unaffected.}

\REFTHR{A structural reason plausibly underlies this near-universal integrality. Fix a single corridor $\ell$: for every path $p$ that uses $\ell$, the occupancy indicator $a_{kp}^{\ell t}$ equals $1$ on a contiguous block of time periods ${t_0,\dots,t_0+\tau_\ell-1}$ and $0$ elsewhere, so the capacity-constraint submatrix restricted to corridor $\ell$'s rows is an interval matrix and therefore totally unimodular. The full path-packing constraint matrix stacks these interval blocks across every corridor a path may traverse, and it is exactly this cross-corridor coupling that can break total unimodularity in general. This suggests, without proving, that the fractional LP optima observed above arise from degenerate ties between routes sharing several corridors rather than from a generic breakdown of integrality, consistent with the ten- and three-way ties identified.}

\subsection{\DBCG versus \ABCG under a Fixed Computational Budget}
\label{sec:exp_singlecut}

To isolate the effect of scenario-wise cut disaggregation, we compare
\DBCG and \ABCG under the same computational
budget. \REFTWO{Both variants solve the same LP-relaxed two-stage problem and therefore
share the same optimal solution upon convergence; differences in Table~\ref{tab:exp1_singlecut}
reflect only the convergence speed under a finite budget --- the run terminates at whichever of (i) optimality gap $\le 10^{-3}$ (relative), (ii) iteration count $\ge 600$, or (iii) wall-clock $\ge 7200$\,s is reached first.}

Table~\ref{tab:exp1_singlecut} compares the proposed \DBCG with the \ABCG under two-hour computational budget. Across all tested instances, \DBCG consistently outperforms \ABCG, yielding objective improvements of 6.48\%, 12.55\%, and 15.45\% for networks with 20, 25, and 30 nodes, respectively. Moreover, \DBCG achieves substantially tighter optimality gaps, reducing the remaining Benders gap from 13.52\% to 3.85\% for 20-node instances, from 21.76\% to 4.19\% for 25-node instances, and from 37.30\% to 22.10\% for 30-node instances. These results demonstrate that disaggregating cuts by scenario significantly strengthens the master problem approximation, leading to both higher-quality solutions and faster convergence within the same computational budget.

\begin{table}[!htbp]
\centering
\caption{\DBCG versus \ABCG under a fixed computational budget.}
\label{tab:exp1_singlecut}
\begin{tabular}{cccccc}
\toprule
&
&
&
&
\multicolumn{2}{c}{Gap (\%)}
\\
\cmidrule(lr){5-6}
$|\mathcal V|$
& \DBCG
& \ABCG
& Loss (\%)
& \DBCG
& \ABCG
\\
\midrule
20 & 130.40 & 121.95 & 6.48 & 3.85 & 13.52 \\
25 & 149.44 & 130.68 & 12.55 & 4.19 & 21.76 \\
30 & 147.48 & 124.69 & 15.45 & 22.10 & 37.30 \\
\bottomrule
\end{tabular}
\end{table}

\end{document}